\documentclass[aps,pra,floatfix,singlecolumn,nofootinbib]{revtex4-1}

\usepackage{graphicx}
\usepackage{subcaption}
\usepackage{amssymb,amsmath,amsfonts,amsthm,mathtools,mathrsfs,eufrak,mathdots}
\usepackage{braket,bm}
\usepackage{epsf,nicefrac}
\usepackage{color}
\usepackage[colorlinks = true]{hyperref}
\usepackage{tikz}
\usetikzlibrary{positioning,shapes.geometric}
\usepackage{pgfplots}
\pgfplotsset{compat=newest}
\pgfplotsset{width=\textwidth}

\newcommand{\abs}[1]{\left| #1 \right|}
\DeclareMathOperator{\arccot}{arccot}

\DeclareMathOperator{\ident}{id}

\def\dim{\mathsf{dim}}

\newcommand{\coll}[1]{\left\{ #1 \right\}}
\def\d{\textrm{d}}

\def\<{\langle}
\def\>{\rangle}
\newcommand{\iu}{\mathrm{i}}

\def\N#1{\left|\!\left|#1\right|\!\right|}

\newcommand{\uv}{\mathbf{1}}

\DeclareMathOperator{\Span}{span}

\def\Rationals{\mathbb Q}
\def\Reals{\mathbb R}

\def\Intgrs{\mathbb Z}

\def\Tr{\operatorname{Tr}}

\def\Tr{\operatorname{Tr}}

\definecolor{darkgreen}{RGB}{0, 150, 0}

\newtheorem{lemma}{Lemma}[section]
\newtheorem{remark}{Remark}[section]

\begin{document}
\title{Opinion bubbles in the quantum Buridan's ass problem}
\author{L. Chaichenets$^1$, J. Novotn\'y$^1$, I. Jex$^1$, S. M. Barnett$^2$, T. Kiss$^3$}
\affiliation{$^1$Department of Physics, Faculty of Nuclear Sciences and Physical Engineering, Czech Technical
University in Prague, B\v rehov\'a 7, 115 19 Praha 1 - Star\'e M\v esto, Czech Republic}
\affiliation{$^2$School of Physics and Astronomy, University of Glasgow, Glasgow G12 8QQ, UK}
\affiliation{$^3$HUN-REN Wigner Research Centre for Physics, 1525 P.O. Box 49, Hungary}
\date{\today}

\begin{abstract}
The quantum Buridan's ass problem concerns a network of interconnected nodes in which transitions between the nodes are governed by the laws of quantum mechanics.  In such situations interference between probability amplitudes plays a key role in determining whether the probability of leaving the initial node is large or small and also the speed at which the associated transitions occur.  In this paper we extend this idea to consider 
transitions between one given set of coupled nodes and the remainder of the network. We have in mind a quantum analogy to changing an opinion, in which the initially occupied set of nodes is one opinion bubble, with the remaining parts of the network corresponding to a different opinion.  We find that quantum interference again plays a crucial role in the dynamics.  A further key component is the role of symmetry between the nodes and their interconnections, as it is symmetry that lies at the heart of the Buridan's ass puzzle. In the classical Buridan's ass paradox it is symmetry between competing alternatives that freezes the ass into indecision.  In the quantum case, however, interference results in a
far richer dynamics.
\end{abstract}

\maketitle

\section{Introduction}
\label{sec:introduction}
Let us begin by outlining a classical (and familiar) analogue of the problem we seek to address.  Sitting in a nice restaurant we find that among the tempting dishes on offer there are two, dish $A$ and dish $B$, that we find equally appealing.  Who has not been in this situation, mentally opting for one dish then the other until an external perturbation, perhaps the arrival of a waiter, forces a choice \cite{Schwartz, Weintraub, Rescher, Lamport, Malhotra}?  This is an every-day example of the Buridan's ass problem in which the unlucky ass starves for want of a reason to choose between two identical piles of carrots \cite{Weintraub,vanFraassen, Curie, Chalmers}.  Yet choosing dish $A$ or dish $B$ may be only the beginning of the diner's problems, for having opted, say, for dish $A$ there are further options that go with it: potatoes or rice, salad or vegetables, water, wine or beer, ... .  Thus option $A$ may, and usually will, have further choices.  The complete set of sub-options corresponding to choosing dish $A$ may be thought of as a composite choice, a decision or opinion bubble bearing in it a clear sign of symmetry breaking .

We can represent all of the choices as a network of interconnected nodes with the connections between them associated with the thought processes, pros and cons, associated with a change of mind \cite{Estrada}.  Our task in this paper is to analyze the dynamics of such
a network within the context of quantum theory.  In a recent publication we addressed the simpler problem of the evolution from a single node to any of a selection of other nodes \cite{Buridan_paper}.  This would correspond, for our diner, to making a decision, or more precisely, a change of mind from dish $A$ to any of the others available without the subsequent additional options.  We found that, in stark contrast to the classical situation,
quantum interferences could greatly speed up this process or, effectively, freeze it depending on the details of the network.

Our aim in this paper is to extend our earlier work \cite{Buridan_paper} to examine the behaviours of interconnected decision bubbles in quantum theory.  To be specific, we 
consider a network or graph of interconnected nodes which are separated into two subgraphs, or
bubbles, $A$ and $B$, with both internal connections within each bubble and also between 
the nodes of bubbles $A$ and $B$.  Our aim is to explore the roles of these connections 
in either enabling or inhibiting transitions between the bubbles.

The quantum network \cite{Weihs, TSJ, Torma} analyzed here has broad potential implications for a diverse range of quantum phenomena. In quantum technologies, the sharing of quantum information between multiple parties is a natural example \cite{Townsend,Phoenix}. Related, also, is the idea of distributed quantum computation \cite{Caleffi} in which strings of qubits must be exchanged, coherently, between multiple processors. Elsewhere in physics, there exists a multitude of important examples involving multiple interfering paths. A few examples serve to illustrate the point: In superconductor quantum information processing, networks of Josephson junctions have been proposed for encoding and protecting individual qubits 
\cite{Ginossar}.  In image processing, the multimode light comprising an image may be filtered, modified, or enhanced by actions in the Fourier plane \cite{Hecht}. In photonics, optical signals may be multiplexed and processed using, for example, waveguide arrays \cite{Saleh}.  Finally, in the multiphoton excitation of atoms, interferences are to be found in the interferences between excitation points that proceed via different electronic energy levels \cite{Krinitzky, Buckle, Shore,Sonja}.  It has also been pointed out that the 
superposition principle may play a role in the all-important process of light-harvesting in photosynthesis \cite{Burgess}.  More exotically, an 
elementary mechanistic model for a quantum brain has been proposed by Fisher \cite{Fisher}.
In particular, the dynamics of the Q-connectone, a three-layer quantum brain model, was
studied with the aim of determining whether quantum coherences might resolve a Buridan 
puzzle in the presence of environmental noise \cite{Wakaura}.

Our paper is organized as follows: our network model is introduced in Section \ref{sec:model}. In Section \ref{sec:spectral_analysis} we solve the associated spectral problem. These results are used in Section \ref{sec:subgraph_effects} devoted to the analysis of effects arising from the network dynamics on the level of subgraphs. Similarly, effects arising on the level of nodes are commented on in detail in Section \ref{sec:node_effects}. Some generalizations are discussed in Section \ref{sec:generalization}. Summary and conclusions are left for Section \ref{sec:conclusions}. 

\section{Network model}
\label{sec:model}
We consider a graph or network comprised of a selection of nodes coupled by lines or edges.
In Fig. \ref{fig:buridan_examples} we depict three such arrangements studied in our earlier 
work \cite{Buridan_paper}.  In each case the initial state of the network is that a single 
quantum excitation is isolated at (one of) the red node(s).  The couplings between the nodes
correspond to transitions mediated by matrix elements of a model Hamiltonian.  For simplicity,
and also to explore the role of symmetry, essential in the Buridan's ass problem, we 
treated each of the links between the nodes as equal, so that the Hamiltonian matrix elements 
between any two linked nodes were the same.  To summarize our earlier results, in the first
network the transition from the initial red node happens ever more quickly as additional nodes 
are added, while in the second, the state is, to a large extent, frozen at the initial node.
In the third there is rapid a transition from one of the red nodes to the other even though 
there is no direct link between them.  In each case, these behaviours are a consequence
of quantum coherence or interference.  Such behaviour is ubiquitous in both optical and
quantum systems and, indeed, may be thought of as a characteristic of quantum dynamics
\cite{Feynman,FeynmanThesis,Hibbs}.

Although our study can apply to a very wide variety of physical implementations, we find it helpful to have in mind just one. This is the arrangement of optical multiports in which light is distributed by means of a collection of beam splitters that play the role of our nodes. Such passive optical networks are quite popular among physicists, both as theoretical tools as well as experimental devices \cite{Reck, interferometers, beamsplitters}. While the topology \cite{Estrada} of the arrangement plays a central role, it is not obvious or straightforward to deduce the form of the associated dynamics, and hence there is quite a space for surprises and counterintuitive results. The analysis of multiports can be done in different ways and a particularly efficient one is to assign a Hamiltonian to the dynamics which such a multiport implements. It is quite natural to expect that the Hamiltonian will exhibit the same symmetry as the multiport itself. In fact, such an approach was followed and the Hamiltonian model of a fully symmetric multiport was defined in \cite{JSZ}. For passive networks the Hamiltonian is of the form
\begin{align}
\label{eqn:hamiltonian}
H & = \sum_{i, j} c_{ij} a_i^\dagger a_j,
\end{align}
where $a_j$ ($a_i^\dagger$) is the annihilation (creation) operator for the mode $i$ (respectively $j$), and $c_{ij}$ is the weight of the line connecting node $j$ to node $i$. 
To ensure Hermiticity, we require the coupling coefficients to satisfy the condition $c_{ji} = \overline{c}_{ij}$.

In this paper we extend these ideas to consider the effects of quantum interference in the transitions between bubbles, a collection of coupled nodes also coupled to the rest of the network. Motivated by the Buridan paradox, we retain a large element of symmetry.  To this 
end, while maintaining an adequate level of generality, we consider the case of our bubbles being two regular subgraphs. Regularity means that each vertex in a graph has the same degree, so that each node has the same number of neighbours. In a regular subgraph, $A$, we shall denote this number of neighbours (not counting the node as its own neighbour) by $d_A$. Let us illustrate this point with a few examples of both regular and irregular subgraphs.  Cliques are fully connected subgraphs, so $d_A = m - 1$, where $m$ is the number of vertices. Cycles have their vertices arranged in a closed loop and hence $d_A = 2$. Star graphs, however, are not regular for $m > 2$, because the hub node has $m - 1$ neighbours but each leaf is connected only to the hub node. A simple example of a graph we were not able to consider previously, but which lends itself to analysis in our new framework is depicted in Fig. \ref{fig:clique_to_cycle}. The two red nodes form one subgraph while the remaining four nodes comprise the second. The intra-subgraph connections within these bubbles are represented, respectively, by the red and blue lines, while the inter-subgraph links are represented by black lines.

We consider the evolution only of a single photon prepared initially in one of the nodes. This corresponds to an initial condition in which a single quantum is transferred through the network.  We remark that our simple model Hamiltonian has two additional advantages that may simplify both experimental realization and also the exploration of multi-excitation 
dynamics. For the first of these we note that, in full generality, the intensities registered in a laser implementation are in proportion to the single-photon probabilities, so it may not be necessary to realize a single-photon experiment \cite{Loudon,Paulbook,SInglePhot}. For the second, we remark that the quadratic form of our Hamiltonian means that it is 
straightforward to obtain the eigenmodes of the network, which enables the full solution of the Heisenberg equations of motion. Here, however, we restrict our attention to the evolution of a single photon in the Schr\"{o}dinger picture.

Our Hamiltonian commutes with the total number of photons $\sum_{i} a_i^\dagger a_i$, and it follows that the relevant state space is spanned by the set of single-photon states $\ket{i} \coloneqq a_i^\dagger \ket{0}$, where $\ket{0}$ is the multimode vacuum
(zero photon) state.  It then follows that the matrix $G \coloneqq (c_{ij})$ is the representation matrix of $H$ in the orthonormal basis $(\ket{i})_i$, and we interpret $G$ as a weighted adjacency matrix of an undirected graph.  We divide the state space into two subspaces, one associated with each of the bubbles. Let the first $m$ single-photon states, $\ket{1}, \ldots, \ket{m}$, span the Hilbert space $\mathcal{H_A}$ associated with Bubble $A$ and the following $n$, i.e. $\ket{m + 1}, \ldots, \ket{m + n}$, span the Hilbert space $\mathcal{H}_B$ associated with Bubble $B$. The full state space of dimension $m + n$ is $\mathcal{H} \coloneqq \mathcal{H}_A \oplus \mathcal{H}_B$.

We complete the description of our model by introducing the coupling strengths.  All onsite energies $c_{ii}$, which only rescale time, are set to unity. All interaction strengths within $A$ are equal to $a$ and all interaction strengths with $B$ are equal to $b$. Finally, the two subgraphs are fully interconnected with a coupling strength $c$. The resulting weighted adjacency matrix is given by
\begin{align}
\label{eqn:adjacency_matrix}
G & =
\left( 
\begin{array}{cc}
  G_A & \begin{array}{cccc}
       c & c & \ldots & c \\
       c & c & \ldots & c \\
       \vdots & \vdots & & \vdots \\
       c & c & \ldots & c
   \end{array} \\
   \begin{array}{cccc}
       c & c & \ldots & c \\
       c & c & \ldots & c \\
       \vdots & \vdots & & \vdots \\
       c & c & \ldots & c
   \end{array} & G_B
\end{array}
\right)
\begin{array}{l}
       \left. \vphantom{\begin{array}{cccc}
       c & c & \ldots & c \\
       c & c & \ldots & c \\
       \vdots & \vdots & & \vdots \\
       c & c & \ldots & c 
       \end{array}} \right\} m \\
       \left. \vphantom{\begin{array}{cccc}
         c & c & \ldots & c \\
       c & 1 & \ldots & c \\
       \vdots & \vdots & & \vdots \\
       c & c & \ldots & c \end{array}} \right\} n
\end{array},
\end{align}
where $G_A$ and $G_B$ are the weighted adjacency matrices of the respective subgraphs.

To connect the situation described above to previous examples, let us note that the star graph from Figure \ref{subfig:star_graph} is obtained for $m = 1$, $n = 6$, $a = c = 1$, $b = 0$, $d_A = 0$ and arbitrary $d_B \in \coll{1, \ldots, 5}$; the fully connected graph from Figure \ref{subfig:full_graph} is represented by $m = 1$, $n = 6$, $a = b = c = 1$, $d_A = 0$ and $d_B = 5$; the missing-link scenario in Figure \ref{subfig:missing_link} has $m = 2$, $n = 5$, $a = 0$, $b = c = 1$, $d_B = 4$ and $d_A \in \coll{0, 1}$. The example from Figure \ref{fig:clique_to_cycle}, a single-edge graph coupled to a $4$-cycle, has $m = 2$, $n = 4$, $d_A = 1,$ $d_B = 2$ (the coupling strengths $a, b, c$ have not been specified).

We shall refer to the collection $m$, $n$, $d_A$, $d_B$, $a$, $b$, $c$ as graph parameters.

\begin{figure}[t]
\centering
\begin{subfigure}{0.32\textwidth}
\centering
\begin{tikzpicture}
\usetikzlibrary{positioning,shapes.geometric} 
\node[draw=none,minimum size=4cm,regular polygon,regular polygon sides=7] (a) {};
\foreach \x in {2, ..., 7} 
{\draw[-,thick,black] (a.corner 1)--(a.corner \x);}
\fill[fill = red] (a.corner 1) circle[radius=5pt]; 
\foreach \x in {2, ..., 7} 
\fill[fill = blue] (a.corner \x) circle[radius=5pt];
\end{tikzpicture}
\caption{Star graph with an excitation initially placed in the hub (red node). The dynamics are full oscillations between the hub node and the uniform superposition on the rim (blue nodes).}
\label{subfig:star_graph}
\end{subfigure}
\begin{subfigure}{0.32\textwidth}
\centering
\begin{tikzpicture}
\usetikzlibrary{positioning,shapes.geometric} 
\node[draw=none,minimum size=4cm,regular polygon,regular polygon sides=7] (a) {};
\foreach \x in {2, ..., 7} 
{\draw[-,thick,blue] (a.corner 1)--(a.corner \x);}
\foreach \x in {3, ..., 7} 
{\draw[-,thick,blue] (a.corner 2)--(a.corner \x);}
\foreach \x in {4, ..., 7} 
{\draw[-,thick,blue] (a.corner 3)--(a.corner \x);}
\foreach \x in {5, ..., 7} 
{\draw[-,thick,blue] (a.corner 4)--(a.corner \x);}
\foreach \x in {6, ..., 7} 
{\draw[-,thick,blue] (a.corner 5)--(a.corner \x);}
\foreach \x in {7, ..., 7} 
{\draw[-,thick,blue] (a.corner 6)--(a.corner \x);}
\fill[fill = red] (a.corner 1) circle[radius=5pt];
\foreach \x in {2, ..., 7} 
\fill[fill = blue] (a.corner \x) circle[radius=5pt];
\end{tikzpicture}
\caption{Fully connected graph with an excitation initially placed in the red node. Asymptotically, with increasing number of blue nodes, the excitation is trapped in the initial node.}
\label{subfig:full_graph}
\end{subfigure}
\begin{subfigure}{0.32\textwidth}
\centering
\begin{tikzpicture}
\usetikzlibrary{positioning,shapes.geometric} 
\node[draw=none,minimum size=4cm,regular polygon,regular polygon sides=7] (a) {};
\foreach \x in {2, ..., 3} 
{\draw[-,thick,blue] (a.corner 1)--(a.corner \x);}
\foreach \x in {4, ..., 5} 
{\draw[-,thick,black] (a.corner 1)--(a.corner \x);}
\foreach \x in {6, ..., 7} 
{\draw[-,thick,blue] (a.corner 1)--(a.corner \x);}
\foreach \x in {3, ..., 3} 
{\draw[-,thick,blue] (a.corner 2)--(a.corner \x);}
\foreach \x in {4, ..., 5} 
{\draw[-,thick,black] (a.corner 2)--(a.corner \x);}
\foreach \x in {6, ..., 7} 
{\draw[-,thick,blue] (a.corner 2)--(a.corner \x);}
\foreach \x in {4, ..., 5} 
{\draw[-,thick,black] (a.corner 3)--(a.corner \x);}
\foreach \x in {6, ..., 7} 
{\draw[-,thick,blue] (a.corner 3)--(a.corner \x);}
\foreach \x in {4, ..., 5} 
{\draw[-,thick,black] (a.corner 6)--(a.corner \x);}
\foreach \x in {7, ..., 7} 
{\draw[-,thick,blue] (a.corner 6)--(a.corner \x);}
\foreach \x in {4, ..., 5} 
{\draw[-,thick,black] (a.corner 7)--(a.corner \x);}
\draw[-,thick,red!20!white] (a.corner 4)--(a.corner 5); 
\foreach \x in {1, ..., 3} 
\fill[fill = blue] (a.corner \x) circle[radius=5pt];
\foreach \x in {4, ..., 5} 
\fill[fill = red] (a.corner \x) circle[radius=5pt];
\foreach \x in {6, ..., 7} 
\fill[fill = blue] (a.corner \x) circle[radius=5pt];
\end{tikzpicture}
\caption{Fully connected graph without one link (light red). Excitation is initially in one of the red nodes. Increasing number of blue nodes, leads to fuller oscillations between the red nodes.}
\label{subfig:missing_link}
\end{subfigure}
\caption{Three graphs considered in our previous work \cite{Buridan_paper}. There, all coupling strengths were equal. In the current paper this situation is generalized to two regular subgraphs (red and blue nodes) having uniform coupling strength within themselves (red and blue edges) and uniformly coupled to one another (black edges). This allows us to cover a number of interesting situations beyond those considered originally.}
\label{fig:buridan_examples}
\end{figure}
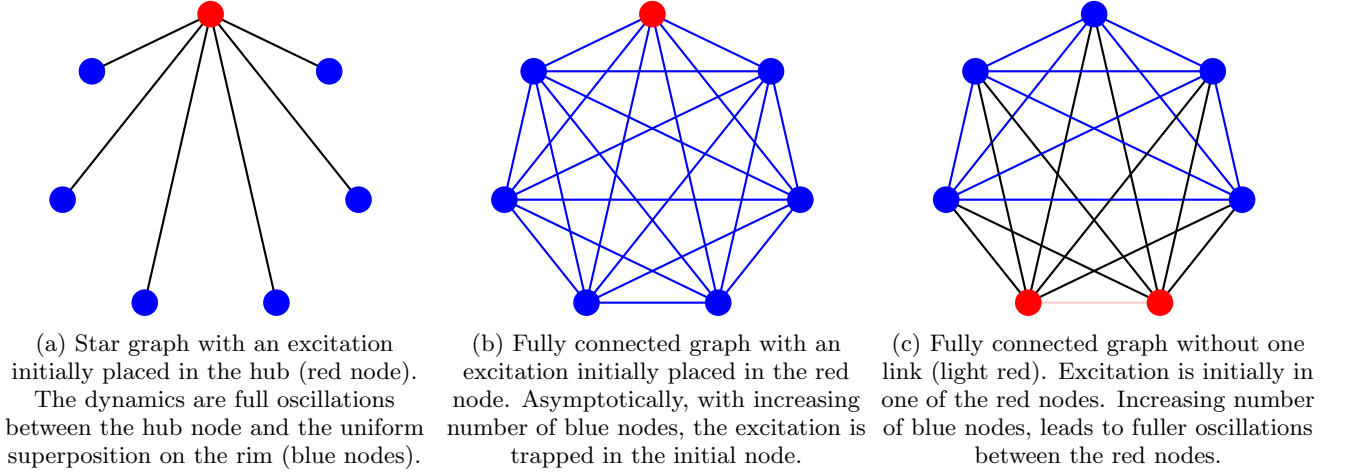

\begin{figure}
\centering
\begin{tikzpicture}[
    scale=1,
    every node/.style={circle, draw=none, minimum size=10pt, inner sep=0pt}
]

\node[fill=red] (r1) at (-2,1) {};
\node[fill=red] (r2) at (-2,-1) {};

\node[fill=blue] (b1) at (2,2) {};
\node[fill=blue] (b2) at (4,0) {};
\node[fill=blue] (b3) at (2,-2) {};
\node[fill=blue] (b4) at (0,0) {};

\draw[-, red, thick] (r1) -- (r2);

\draw[-, blue, thick] (b1) -- (b2) -- (b3) -- (b4) -- (b1);

\draw[-, black, thick] (r1) -- (b1);
\draw[-, black, thick] (r1) -- (b2);
\draw[-, black, thick] (r1) -- (b3);
\draw[-, black, thick] (r1) -- (b4);

\draw[-, black, thick] (r2) -- (b1);
\draw[-, black, thick] (r2) -- (b2);
\draw[-, black, thick] (r2) -- (b3);
\draw[-, black, thick] (r2) -- (b4);
\end{tikzpicture}
\caption{An example graph comprised of two regular subgraphs. The red subgraph on the left consists of two connected vertices. The blue subgraph on the right is a cycle consisting of four nodes. The two subgraphs are fully interconnected with black edges.}
\label{fig:clique_to_cycle}
\end{figure}
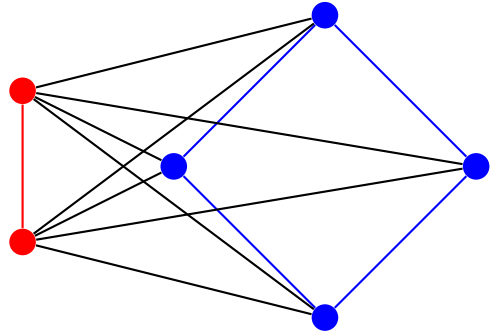

\section{Spectral analysis and dynamics}
\label{sec:spectral_analysis}
In order to analyze the evolution of the excitation within the network we need to explore the spectral properties of $G$. As $A$ is regular, $G_A$ has the normalized uniform vector $\ket{\phi} \coloneqq \frac{1}{\sqrt{m}} \uv_m$ as an eigenvector to the eigenvalue $\lambda_A \coloneqq d_A a + 1$ --- we shall refer to such eigenvector-eigenvalue pairs as uniform eigenpairs. Let us fix the notation and complement $\ket{\phi}$ with $\ket{\phi_i}$, $i \in \coll{1, \ldots, m - 1}$, to obtain an orthonormal basis of eigenvectors of $G_A$ corresponding the eigenvalues $\mu_1, \ldots, \mu_{m - 1}$. Because of this orthogonality, the sum of components of each $\ket{\phi_i}$ is zero and hence, just by extending $\ket{\phi_i}$ with zeros, we immediately obtain orthogonal eigenvectors $\ket{\psi_i} \coloneqq (\phi_i, 0, \ldots, 0)$ of $G$ corresponding to the eigenvalues $\lambda_i \coloneqq \mu_i$, $i \in \coll{1, \ldots, m - 1}$. We also extend $\ket{\phi}$ to $\ket{\psi_A} \coloneqq (\phi, 0, \ldots, 0)$, noting that it is \emph{not} an eigenvector of $G$, unless $c = 0$.

The same analysis applies to Subgraph $B$: Its weighted adjacency matrix $G_B$ has the uniform eigenpair consisting of $\ket{\chi} \coloneqq \frac{1}{\sqrt{n}} \uv_n$ and $\lambda_B \coloneqq d_B b + 1$. The remaining $n - 1$ eigenvectors $\ket{\chi_j}$ of $G_B$ extend to the eigenvectors $\ket{\psi_{j + m + 1}} \coloneqq (0, \ldots, 0, \chi_j)$ of $G$ corresponding to the eigenvalues $\lambda_{j + m + 1} \coloneqq \nu_j$, $j \in \coll{1, \ldots, n - 1}$, of $G_B$. Again, we extend $\ket{\chi}$ to $\ket{\psi_B} \coloneqq (0, \ldots, 0, \chi)$.

Because all entries in \eqref{eqn:adjacency_matrix} are real and because $\ket{\psi_1}, \ldots, \ket{\psi_{m - 1}}, \ket{\psi_{m + 2}}, \ldots, \ket{\psi_{m + n}}$ are orthogonal, the remaining two orthogonal eigenvectors of $G$ must be a rotation of $(\ket{\psi_A}, \ket{\psi_B})$. We therefore make the ansatz
\begin{align}
\label{eqn:ansatz}
\ket{\psi} & = \gamma \ket{\psi_A} + \ket{\psi_B}.
\end{align}
The eigenvalue problem $H \ket{\psi} = \lambda \ket{\psi}$ then reduces to the system of equations
\begin{align}
\label{eqn:eigenvalue_problem}
\gamma \lambda_A + \tilde{c} & = \lambda \gamma, &
\tilde{c} \gamma + \lambda_B & = \lambda,
\end{align}
where $\tilde{c} \coloneqq c \sqrt{m n}$. We obtain the two remaining eigenvalues
\begin{align}
\label{eqn:perturbed_eigenvalues}
\lambda_{\pm} & =
\frac{\lambda_A + \lambda_B}{2} \pm \tilde{c} \sqrt{1 + \eta^2}, &
\eta & \coloneqq \frac{\lambda_A - \lambda_B}{2 \tilde{c}},
\end{align}
and their associated normalized eigenvectors in the form
\begin{align}
\label{eqn:perturbed_eigenvectors}
\ket{\psi_{\pm}} & =
\frac{\gamma_\pm}{\sqrt{1 + \gamma_\pm^2}} \ket{\psi_A} +
\frac{1}{\sqrt{1 + \gamma_\pm^2}} \ket{\psi_B}, &
\gamma_\pm & = \eta \pm \sqrt{1 + \eta^2}.
\end{align}
For a more succinct presentation we will sometimes denote $\lambda_m \coloneqq \lambda_+$, $\lambda_{m + 1} \coloneqq \lambda_-$, $\ket{\psi_m} \coloneqq \ket{\psi_+}$, and $\ket{\psi_{m + 1}} \coloneqq \ket{\psi_-}$.

The above construction leads directly to two simplifications. The first is the pair of identities:
\begin{align}
\label{eqn:scalar_products}
C & \coloneqq
\braket{\psi_A | \psi_+} =
\frac{\gamma_+}{\sqrt{1 + \gamma_+^2}} =
\frac{1}{\sqrt{1 + \gamma_-^2}} = \braket{\psi_B | \psi_-}, &
S & \coloneqq
\braket{\psi_B | \psi_+} = \frac{1}{\sqrt{1 + \gamma_+^2}} =
\frac{-\gamma_-}{\sqrt{1 + \gamma_-^2}} = -\braket{\psi_A | \psi_-},
\end{align}
i.e. $\ket{\psi_+}$ indeed results from the rotation of $\ket{\psi_A}$ in the $\ket{\psi_A}$-$\ket{\psi_B}$-plane by the angle $\varphi \coloneqq \arccot(\gamma_+)$. A corresponding statement is true for $\ket{\psi_-}$ and $\ket{\psi_B}$. A diagrammatic representation of these
properties is given in Figure \ref{fig:rotated_eigenvectors}.
\begin{figure}
\centering
\begin{tikzpicture}[x=0.0833\linewidth, y=0.0833\linewidth]
\pgfmathsetmacro{\GammaPlus}{5./4.}
\pgfmathsetmacro{\GammaMinus}{-1. / \GammaPlus}
\pgfmathsetmacro{\CosPhi}{\GammaPlus / veclen(1., \GammaPlus)}
\pgfmathsetmacro{\SinPhi}{1. / veclen(1., \GammaPlus)}
\pgfmathsetmacro{\RotationAngle}{atan(\SinPhi/\CosPhi)}

\pgfmathsetmacro{\AxesScale}{4.}
\pgfmathsetmacro{\TicksHeight}{.1}
\pgfmathsetmacro{\XMax}{6.}
\pgfmathsetmacro{\XMin}{-6.}
\pgfmathsetmacro{\YMax}{6.}
\pgfmathsetmacro{\YMin}{-1.}

\pgfmathsetmacro{\OriginX}{0.}
\pgfmathsetmacro{\OriginY}{0.}
\pgfmathsetmacro{\CosPhiCanvas}{\CosPhi * \AxesScale}
\pgfmathsetmacro{\SinPhiCanvas}{\SinPhi * \AxesScale}
\pgfmathsetmacro{\GammaPlusCanvas}{\GammaPlus * \AxesScale}
\pgfmathsetmacro{\GammaMinusCanvas}{\GammaMinus * \AxesScale}

\coordinate (origin) at (\OriginX, \OriginY);
\coordinate (psi_a) at (\AxesScale, 0.);
\coordinate (psi_b) at (0., \AxesScale);
\coordinate (psi_plus) at (\CosPhiCanvas, \SinPhiCanvas);
\coordinate (psi_minus) at (-\SinPhiCanvas, \CosPhiCanvas);
\coordinate (projected_psi_plus) at (\GammaPlusCanvas, \AxesScale);
\coordinate (projected_neg_psi_plus) at (-\CosPhi / \SinPhi, \YMin);
\coordinate (projected_psi_minus) at (\GammaMinusCanvas, \AxesScale);

\draw [->, ultra thin] (\XMin, \OriginY) -- (\XMax, \OriginY);
\draw [->, ultra thin] (\OriginX, \YMin) -- (\OriginX, \YMax);

\draw [-, thick] (\GammaMinusCanvas, \TicksHeight) -- (\GammaMinusCanvas, -\TicksHeight) node[below]{$\gamma_-$};
\draw [-, thick] (-\SinPhiCanvas, \TicksHeight) -- (-\SinPhiCanvas, -\TicksHeight) node[below]{$-S$};
\draw [-, thick] (\CosPhiCanvas, \TicksHeight) -- (\CosPhiCanvas, -\TicksHeight) node[below right]{$C$};
\draw [-, thick] (\GammaPlusCanvas, \TicksHeight) -- (\GammaPlusCanvas, -\TicksHeight) node[below]{$\gamma_+$};

\draw [-, thick] (-\TicksHeight, \SinPhiCanvas) -- (\TicksHeight, \SinPhiCanvas) node[anchor=south west, above right]{$S$};
\draw [-, thick] (\TicksHeight, \CosPhiCanvas) -- (-\TicksHeight, \CosPhiCanvas) node[ anchor=south east, above left]{$C$};

\usetikzlibrary{angles, quotes}
\draw pic ["$\varphi$", draw, -latex, angle radius=2.5cm, angle eccentricity=.75] {angle = psi_a--origin--psi_plus};
\draw pic [".", draw, thick, angle eccentricity=.5] {right angle = psi_minus--origin--projected_neg_psi_plus};

\draw [-, dashed, ultra thin] (\XMin, \AxesScale) -- (\XMax, \AxesScale);
\draw [-, dotted, thin] ( \OriginX, \SinPhi * \AxesScale) -- (psi_plus);
\draw [-, dotted, thin] ( \OriginX, \CosPhi * \AxesScale) -- (psi_minus);

\draw [-, dashed, ultra thin] (\GammaMinus * \AxesScale, \OriginY) -- (\GammaMinus * \AxesScale, \YMax);
\draw [-, dotted, thin] (-\SinPhi * \AxesScale, \OriginY) -- (psi_minus);
\draw [-, dotted, thin] (\CosPhi * \AxesScale, \OriginY) -- (psi_plus);
\draw [-, dashed, ultra thin] (\GammaPlus * \AxesScale, \OriginY) -- (\GammaPlus * \AxesScale, \YMax);

\draw [-, dashed, thick] (psi_plus) -- (projected_psi_plus);
\draw [-, dashed, thick] (psi_minus) -- (projected_psi_minus);
\draw [-, dashed, thick] (origin) -- (projected_neg_psi_plus);

\draw [-, dashed, ultra thin] (\AxesScale, .0) arc(0:180:\AxesScale);

\draw [-latex, thick, blue] (origin) -- (psi_a) node[midway, below]{$\ket{\psi_A}$};
\draw [-latex, thick, blue] (origin) -- (psi_b) node[midway, sloped, above left]{$\ket{\psi_B}$};
\draw [-latex, thick, red] (origin) -- (psi_plus) node[midway, sloped, above right]{$\ket{\psi_+}$};
\draw [-latex, thick, red] (origin) -- (psi_minus) node[midway, sloped, below]{$\ket{\psi_-}$};
\end{tikzpicture}
\caption{Off-diagonal interaction blocks in Equation \eqref{eqn:adjacency_matrix} result in the rotation of the original, orthogonal eigenvectors $\ket{\psi_A}$, $\ket{\psi_B}$ by the angle $\varphi$.}
\label{fig:rotated_eigenvectors}
\end{figure}
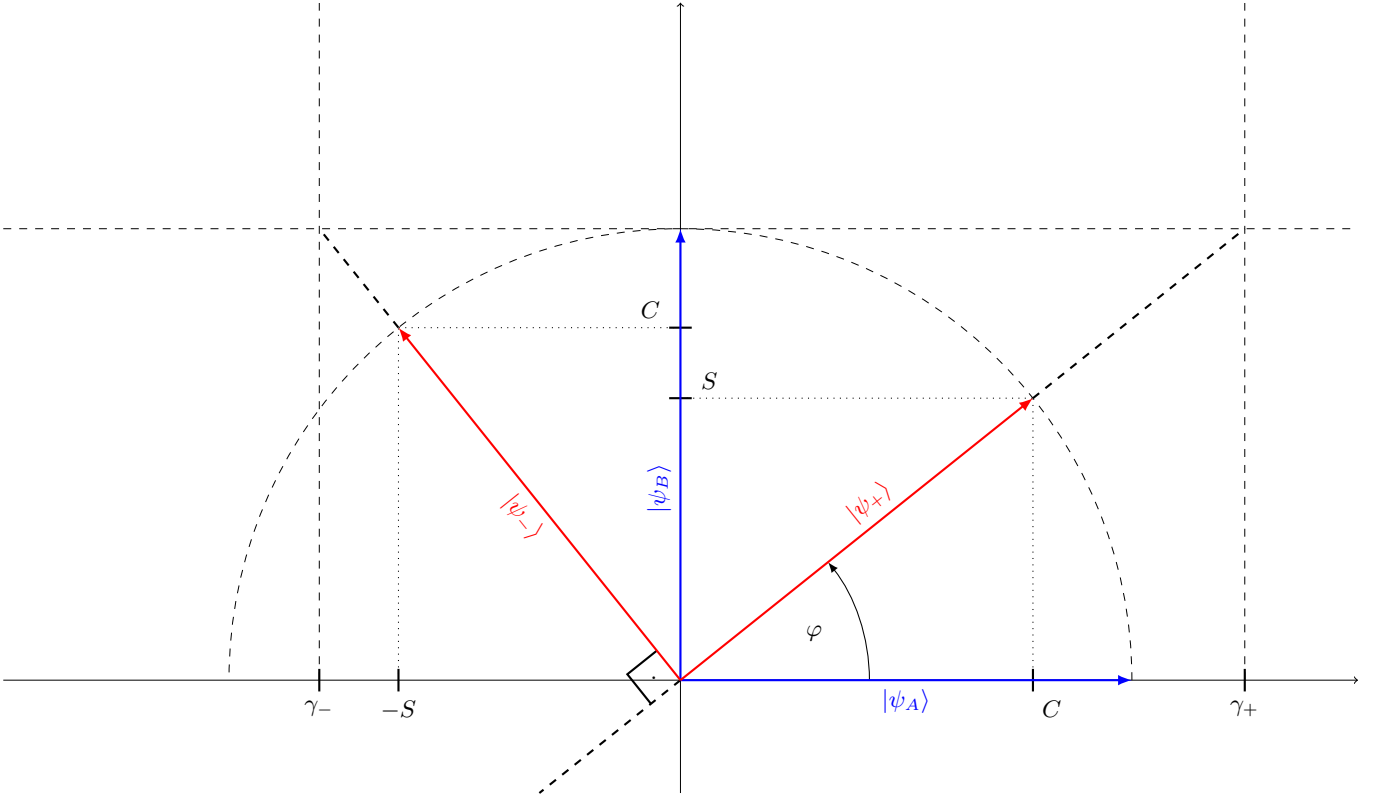
The second one is that while the eigenvalues $\mu_1, \ldots, \mu_{m - 1}$ and the corresponding eigenvectors $\ket{\phi_1}, \ldots, \ket{\phi_{m - 1}}$ depend on the particular choice of Subgraph $A$, the uniform eigenpair $\left(\frac{1}{\sqrt{m}} \uv_m, d_A a + 1\right)$ does not require detailed knowledge of $A$. Moreover, a corresponding statement is true for Subgraph $B$. The second simplification is significant, since, as it will be shown, the uniform eigenpairs alone determine the transport properties between the two subgraphs.

We begin our analysis without assuming any detailed knowledge on the structure of the Subgraphs $A$ and $B$, i.e. having only the graph parameters at our disposal. We recall that pure states evolve via the Schrödinger propagator
\begin{align}
\label{eqn:schroedinger_propagator}
U(t) & = e^{-\iu t H} =
\left(\sum_{k = 1}^{m - 1} e^{-\iu t \lambda_k} \ket{\psi_k} \bra{\psi_k}\right) +
\left(e^{-\iu t \lambda_+} \ket{\psi_+} \bra{\psi_+}\right) +
\left(e^{-\iu t \lambda_-} \ket{\psi_-} \bra{\psi_-}\right) +
\left(\sum_{k = m + 2}^{m + n} e^{-it \lambda_k} \ket{\psi_k} \bra{\psi_k}\right),
\end{align}
where we, as usual, set $\hbar = 1$. Hence, the Hilbert spaces $\mathcal{H}_{A'}$, $\mathcal{H}_{B'}$ given by
\begin{align*}
\mathcal{H}_{A'} & \coloneqq \Span \coll{\ket{\psi_1, \ldots, \psi_{m - 1}}} \subseteq
\Span \coll{\ket{\psi_1}, \ldots, \ket{\psi_{m - 1}}, \ket{\psi_A}} = \mathcal{H}_A, \\
\mathcal{H}_{B'} & \coloneqq \Span \coll{\ket{m + 2}, \ldots, \ket{m + n}} \subseteq 
\Span \coll{\ket{\psi_B}, \ket{\psi_{m + 2}}, \ldots, \ket{\psi_{m + n}}} = \mathcal{H}_B.
\end{align*}
are invariant under the Schrödinger evolution. Let us further fix the notation and introduce the projectors
\begin{align}
\label{eqn:regional_measurements}
M_{A'} & \coloneqq \sum_{k = 1}^{m - 1} \ket{\psi_i} \bra{\psi_i}, &
M_A & \coloneqq M_{A'} + \ket{\psi_A} \bra{\psi_A}; &
M_{B'} & \coloneqq \sum_{k = m + 2}^{m + n} \ket{\psi_i} \bra{\psi_i}, &
M_B & \coloneqq M_{B'} + \ket{\psi_B} \bra{\psi_B}.
\end{align}

Everything is now determined in order to calculate the transfer probability $P_{A \to B}(\rho, t)$ of measuring an excitation in Subgraph $B$ at a time $t$, assuming that the initial mixed-state $\rho$ had support in Subgraph $A$ (that is $\rho = M_A \rho = \rho M_A = M_A \rho M_A$ --- all equalities are equivalent). The density matrix evolves as $\rho(t) = U(t) \rho U(t)^\dagger$ and the desired probability is given by
\begin{align}
\label{eqn:general_transition_probability}
P_{A \to B}(\rho, t) & = \Tr \left[M_B \rho(t)\right] =
\Tr \left[M_B U(t) M_A \rho M_A U(t)^\dagger M_B \right],
\end{align}
where we used the cyclic permutation property of the trace, the fact that $M_B$ is a projection, and that the support of $\rho$ is in $A$ for the second equality. Using the identities \eqref{eqn:regional_measurements} and the orthogonality properties of our basis states we obtain
\begin{align*}
M_B U(t) M_A & =
\left(
\braket{\psi_B | \psi_+} \braket{\psi_+ | \psi_A} e^{-\iu t \lambda_+} + 
\braket{\psi_B | \psi_-} \braket{\psi_- | \psi_A} e^{-\iu t \lambda_-}
\right) \ket{\psi_B} \bra{\psi_A}.
\end{align*}
The ensuing scalar products are given in Display \eqref{eqn:scalar_products}. Inserting the above into Equation \eqref{eqn:general_transition_probability}, and using the cyclic permutation property of the trace gives the transition probability between the subgraphs
\begin{align}
\label{eqn:transition_probability}
P_{A \to B}(\rho, t) & =
4 C^2 S^2 \frac{1 - \cos((\lambda_+ - \lambda_-)t)}{2} \braket{\psi_A | \rho | \psi_A}.
\end{align}
Let us note that the maximal transition probability $P_\text{trans} \coloneqq 4 C^2 S^2$ has the geometric representation $\sin^2(2 \varphi)$, and
\begin{align}
\label{eqn:maximal_transition_probability}
P_\text{trans} & =
\frac{4 \gamma_+^2 \gamma_-^2}
{\left(1 + \gamma_+^2\right)\left(1 + \gamma_-^2\right)} =
\frac{1}{1 + \eta^2},
\end{align}
where we employed a Vieta’s formula $\gamma_+ \gamma_- = -1$ for the last equality. Let us also
introduce the complementary probability 
\begin{align}
\label{eqn:minimal_trapping}
P_\text{trap} & \coloneqq 1 - P_\text{trans} = \frac{1}{1 + \frac{1}{\eta^2}}
\end{align}
as the minimal trapping probability. The reason for introducing it will become apparent in Section \ref{subsec:trapping_subgraphs}. At this point we merely observe that
\begin{align}
\label{eqn:fraction_of_probabilities}
\eta^2 & = \frac{P_\text{trap}}{P_\text{trans}},
\end{align}
i.e. a link between the abstract geometric parameter $\eta$ on the one hand and the minimal trapping and its complementary maximal transition probability on the other hand.

The return probability $P_{A \to A}(\rho, t)$, i.e. the probability of rediscovering an excitation in Subgraph $A$ at a time $t$, is 
\begin{align}
\label{eqn:return_probability}
P_{A \to A}(\rho, t) & = \Tr\left[M_A \rho(t) \right] = \Tr[(\ident - M_B) \rho] =
1 - P_{A \to B}(\rho, t),
\end{align}
which follows just by exploiting the linearity of the trace and the fact that $M_A + M_B = \ident$, where $\ident$ denotes the identity operator.

We will not discuss $P_{B \to A}(\rho, t)$ or $P_{B \to B}(\rho, t)$ as they are obtained by simply swapping $A$ and $B$ in the discussion above. We do note however, that the maximal transition probability is the same for both subgraphs: Swapping the subgraphs flips the sign of $\eta$, but $P_\text{trans}$, as a function of $\eta^2$, is unchanged. The angular frequency
\begin{align}
\label{eqn:angular_frequency}
\omega & \coloneqq \lambda_+ - \lambda_- = \frac{2 \tilde{c}}{\sqrt{P_\text{trans}}},
\end{align}
also merely flips its physically irrelevant sign, when the subgraphs are swapped. Finally, we notice from the above that $\omega$ and $P_\text{trans}$ can be set independently. More precisely we can rescale all coupling constants by the same factor, which will keep $P_\text{trans}$ invariant but adjust $\omega$ to an arbitrary value.

\section{Effects Arising at the Level of Subgraphs}
\label{sec:subgraph_effects}
We are primarily interested in two somewhat complementary effects already observed in \cite{Buridan_paper}: full oscillations and trapping. Full oscillations occur when the excitation initially localized in one of the subgraphs is measured with probability one in the other subgraph at a later point in time. Trapping occurs when an excitation localized in one of the subgraphs, at least in part, stays confined to that subgraph. 

\subsection{Full Oscillations}
\label{subsec:full_oscillations_subgraphs}
Of course, an excitation initially localized in Subgraph $A$ may remain there forever, without any oscillation. This is the case, e.g., for $m \geq 2$ and $\rho = \ket{\psi_1} \bra{\psi_1}$. In general, however, due to Equations \eqref{eqn:return_probability} and \eqref{eqn:transition_probability}, the return probability
\begin{align}
\label{eqn:return_probability_concrete}
P_{A \to A}(\rho, t) & =
1 - P_\text{trans} \frac{1 - \cos(\omega t)}{2} \braket{\psi_A | \rho | \psi_A}
\end{align}
will oscillate with angular frequency $\omega$ between $1$ and its minimal value $P_{A \to A}(\rho) \coloneqq 1 - P_\text{trans} \braket{\psi_A | \rho | \psi_A}$. Full oscillations occur exactly when $P_{A \to A}(\rho,t) = 0$ for some $t$, i.e. when $P_\text{trans} = 1 = \braket{\psi_A | \rho | \psi_A}$. By Lemma \ref{lem:convex_max} the second equality is true if and only if $\rho = \ket{\psi_A} \bra{\psi_A}$. By Equation \eqref{eqn:maximal_transition_probability} the first equality is fulfilled if and only if the geometric parameter $\eta$ vanishes (i.e. when the rotation angle $\varphi$ is $\frac{\pi}{4}$). From $\eta = \frac{d_A a - d_B b}{2 c \sqrt{m n}}$ we can infer that this happens, iff
\begin{align}
\label{eqn:full_oscillations_condition}
d_A a = d_B b.
\end{align}
In this case, by Equation \eqref{eqn:angular_frequency}, the angular frequency of the full oscillations is given by
\begin{align}
\label{eqn:frequency_full_oscillations}
\omega & = 2 c \sqrt{m n}.
\end{align}
In particular, it can be adjusted freely by simply varying the inter-subgraph coupling constant $c$ only. This allows for a, perhaps unexpected, degree of control: we can choose the ratio between the interaction strengths, $\frac{a}{b}$, as $\frac{d_B}{d_A}$ to obtain full oscillations. Furthermore, we can tune $\omega$, by adjusting $c$, to a sufficiently small value. In this way, at least for a certain amount of time, we can keep the excitation almost fully in one of the subgraphs.

Asymptotically, we obtain full oscillations in the limit $c \to \infty$, if all other graph parameters are kept fixed. Let us note also, from Equation \eqref{eqn:angular_frequency}, that among all subgraphs of the same size with the same inter-subgraph coupling strength $c$, the ones exhibiting full oscillations have the lowest oscillation frequency.

\subsection{Trapping}
\label{subsec:trapping_subgraphs}
As mentioned above, trapping is somewhat complementary to full oscillations. While we need the minimal return probability $P_{A \to A}(\rho,t)$ to vanish for some time $t$ to give a full oscillation, trapping occurs if $P_{A \to A}(\rho,t)$ stays away from zero, i.e. a non-vanishing part of the wave function always remains in Subgraph $A$. From Equation \eqref{eqn:return_probability_concrete} we immediately see that complete trapping, i.e. $P_{A \to A}(\rho,t) = 1$, occurs if $\rho$ is supported on the invariant Hilbert space $\mathcal{H}_{A'}$. In contrast, minimal trapping occurs if $\rho = \ket{\psi_A}\bra{\psi_A}$. From Equation \eqref{eqn:return_probability} we immediately obtain
\begin{align}
P_\text{trap} & = 1 - P_\text{trans} \leq P_{A \to A}(\rho, t)
\end{align}
with equality attained if and only if $\rho = \ket{\psi_A} \bra{\psi_A}$ and $t \in \frac{\pi}{\omega} (2 \Intgrs + 1)$. This fact motivates our earlier designation of $P_\text{trap}$ as the minimal trapping probability.

The above shows that minimal trapping, attained for $\rho = \ket{\psi_A} \bra{\psi_A}$, is never a complete trapping. However, it becomes asymptotically complete for $\abs{\eta} \to \infty$ (i.e. for $\varphi \to 0^+$ or $\varphi \to \frac{\pi}{2}^-$). Notice that the geometric parameter $\eta$ can be written as
\begin{align}
\eta & =
\frac{1}{2} \left( \frac{a}{c} \sqrt{\frac{m}{n}} \frac{d_A}{m} -
\frac{b}{c} \sqrt{\frac{n}{m}} \frac{d_B}{n} \right),
\end{align}
revealing a competition between the relative contributions of Subgraphs $A$ and $B$: The term to the left of the minus sign corresponds to Subgraph $A$ and pushes $\varphi$ towards $0$, while the term to the right of the minus sign belongs to Subgraph $B$ and pushes $\varphi$ towards $\frac{\pi}{2}$. Each of these contributions is a product of the relative coupling strength, a size proportion, and the asymptotic subgraph’s density. On the one hand, we obtain $\abs{\eta} \to \infty$, for $a \to \infty$ or $b \to \infty$ with other graph parameters being fixed. On the other hand, for fixed non-zero coupling strengths, we obtain an exact condition for asymptotic complete minimal trapping: One of the subgraphs must be much larger than the other, say $\frac{n}{m} \to \infty$, and this large subgraph must be sufficiently dense, i.e. $\frac{d_B}{\sqrt{m n}} \to \infty$. It is interesting to note that it does not matter which of the subgraphs satisfies this condition. The asymptotic complete minimal trapping will occur in both of them.

As the minimal return probability $P_{A \to A}(\rho)$ varies significantly with the initial state $\rho$, we may characterize trapping by taking the average with respect to all possible initial states supported on Subgraph $A$. To that end, let us consider any unitary-invariant probability measure $\mu$ on the density matrices supported on Subgraph $A$. From Equation \eqref{eqn:return_probability_concrete} we immediately recognize that the return probability $P_{A \to A}(\rho, t)$ is linear in the initial condition $\rho$. By Lemma \ref{lem:unitary_invariant_mean}, we obtain that the average return probability $\overline{P}_{A \to A}(t)$ coincides with the minimal return probability for the initial state $\rho = \frac{1}{m} M_A$, i.e. the maximally mixed state living only on the Hilbert space $\mathcal{H}_A$. We define the average trapping probability $\overline{P}_\text{trap}$ as the minimum in time of $\overline{P}_{A \to A}(t)$, i.e.
\begin{align}
\label{eqn:mean_minimal_trapping}
\overline{P}_{A \to A}(t) & =
P_{A \to A} \left( \frac{1}{m} M_A, t \right) =
1 - \frac{P_\text{trans}}{m} \frac{1 - \cos(\omega t)}{2} \geq 1 - \frac{P_\text{trans}}{m} \eqqcolon \overline{P}_{\text{trap}}.
\end{align}
Let us note that $\overline{P}_\text{trap}$ coincides with the average of the minimal return probabilities for eigenbasis states with complete trapping $P_{A \to A}\left(\ket{\psi_1} \bra{\psi_1}\right) = \ldots = P_{A \to A}\left(\ket{\psi_{m - 1}} \bra{\psi_{m - 1}}\right) = 1$ and the minimal return probability for the transfer state $P_{A \to A}\left(\ket{\psi_A} \bra{\psi_A}\right) = 1 - P_\text{trans}$. We conclude by remarking that $\overline{P}_\text{trap}$ approaches $1$ with increasing size of Subgraph $A$, uniformly in size and structure of Subgraph $B$.

\section{Effects Arising at the Level of Nodes}
\label{sec:node_effects}
In the previous section we discussed the effects of trapping and full oscillations at the level of subgraphs. Nonetheless, we cannot immediately tell from the above whether an excitation inserted into a single node of Subgraph $A$ will fully oscillate in and out of the node, even if we had complete oscillations at the level of subgraphs. Here, we extend our discussion to effects on the level of nodes. To that end, we will need to assume more knowledge of the subgraphs’ structure.

\subsection{Single-Node Subgraphs}
\label{subsec:single-node}
We begin with the readily accessible case in which Subgraph $A$ consists of only a single node. All subgraph-level effects immediately become node-level effects and we can simplify the arising formulae. The only initial state supported on $A$ is $\ket{\psi_A} = \ket{1}$ and the invariant Hilbert space $\mathcal{H}_{A'}$ is trivial. Therefore, full oscillations and complete trapping are exactly complementary --- the higher the (minimal) trapping the less pronounced the oscillations. More precisely, the return probability is given by
\begin{align}
\label{eqn:return_probability_one_node}
P_{1 \to 1}(t) & = 1 - \frac{1}{1 + \eta^2} \frac{1 - \cos(\omega t)}{2}, &
\eta^2 & = \frac{d_B^2 b^2}{4 c^2 n}, &
\omega & = \sqrt{4 c^2 n + d_B^2 b^2}.
\end{align}

Let us stress that the minimal return probability is a monotonic function of the parameter $\eta^2$. In particular, two networks with equal parameter $\eta^2$ exhibit the same minimal trapping probability. For example, we obtain $\eta^2 = 1$, i.e. $P_\text{trap} = \frac{1}{2}$, for $n = 4$, $d_B = 3$, $b = 4$, $c = 3$, and also for $n = 9$, $d_B = 9$, $b = 3$, $c = 4$, albeit the oscillation frequency $\omega = \frac{48}{\sqrt{2}}$ is twice as high in the second case compared to the first one.

In general, the oscillation frequency $\omega = \sqrt{4 c^2 n + d_B^2 b^2}$ can be increased if we either increase the size $n$, the coupling $b$ or the connectivity $d_B$ in Subgraph $B$ or if we increase the coupling strength $c$ between the subgraphs. We once again emphasize that any parameter pair $(\eta^2, \omega) \in \Reals^+ \times \Reals^+$ is achievable by tuning the graph parameters. In particular, the minimal trapping probability does not fix the oscillation frequency. More precisely, the minimal set of parameters is $d_B b$, $c^2 n$ meaning that any two graphs have the same pair $(d_B b, c^2 n)$ iff their return probability functions coincide.

Clearly, with larger $\eta^2$ we get a higher level of trapping. We see from Equation \eqref{eqn:return_probability_one_node} that while the overall mutual interaction among vertices of Subgraph $B$ enhances trapping, an increase in the subgraph’s size or its coupling with the initial vertex tends to facilitate the release of the excitation. It is the interference within Subgraph $B$ which confines an excitation (at least partially) to the only node of Subgraph $A$. 

On the one hand, full oscillations and complete absence of trapping occur iff $d_B = 0$ or $b = 0$, i.e. for $\eta = 0$. This is in full agreement with the result we obtained in \cite{Buridan_paper}: there is no trapping for the star subgraph, which corresponds to $b = 0$, $c = 1$ here. On the other hand, we obtain asymptotically complete trapping for $\abs{\eta} \to \infty$. For a fixed Subgraph $B$ this happens for $\frac{b}{c} \to \infty$. This process is illustrated in Figure \ref{fig:coupling_strength}.

Let us consider also the following limit: We fix the coupling strengths and the family of Subgraph $B$ (e.g. a clique, or a star-shaped graph) but let its size $n$ grow. We can distinguish three cases:

The first one is $d_B \gg \sqrt{n}$ for which we obtain asymptotically complete trapping. This is the case, for example, for a clique because its vertices have degree $n - 1$. Then, in the limit of large $n$,
\begin{align*}
\eta^2 & = \frac{(n - 1)^2 b^2}{4 c^2 n} \approx \frac{n b^2}{4c^2}.
\end{align*}
Consequently, $\eta^2$ increases nearly linearly with the size $n$ of the clique. If we increase $n$ four times, we have to approximately halve the coupling strength $b$ to obtain the same trapping effect. This is depicted in Figure \ref{fig:constant_trapping}.

The second case is $d_B \approx \sqrt{n}$ for which the trapping is asymptotically independent of the graph size $n$. While such graphs do not arise among the usual graph families, they can be constructed easily: For each $n$ take as $d_B$ the integer closest to $\sqrt{n}$ making $d_B n$ even. The last requirement is an intimate relationship between the number of vertices and the degree of a graph and we refer to the handshake Lemma \ref{lem:handshake}. If the choice of $d_B$ is still ambiguous (as for $n = 9$), take the larger integer. Subsequently, construct the $d_B$-regular graph with $n$ nodes as in Remark \ref{rem:handshake}. For such graphs of intermediate density we obtain a size-independent trapping
\begin{align}
\eta^2 & \approx \frac{b^2}{4 c^2}
\end{align}
and sketch a number of return probabilities in Figure \ref{fig:intermediate_graphs}.

The last case is $d_B \ll \sqrt{n}$ for which the trapping probability goes to zero with diverging size $n$. Such behavior is exhibited e.g. in cycle graphs, because each vertex has degree $d_B = 2$ and consequently
\begin{align}
\eta^2 & = \frac{b^2}{c^2 n}.
\end{align}
Therefore, $\eta^2$ is inversely proportional to the size of the cycle. Another example are hypercubes, for which the degree of each vertex is equal to the dimension $d$ of the hypercube and the number of nodes is $2^d$. We obtain
\begin{align}
\eta^2 & = \frac{d^2 b^2}{4 c^2 2^d},
\end{align}
which decreases exponentially with the dimension $d$ of the hypercube.

Finally, we note that the oscillation frequency necessarily diverges with increasing $n$ in all of the three cases discussed above. Of course, for a fixed graph, one can keep the frequency constant and vary the minimal trapping probability by adjusting the coupling strengths $b$, and $c$. This behaviour is illustrated in Figure \ref{fig:constant_frequency}

\begin{figure}
    \centering
    \begin{tikzpicture}
\begin{axis}[
    width=.5\linewidth,
    axis lines=left,
    xlabel = {Time},
    ylabel = {Probability},
    xtick = {.0691, 0.2856, 0.4778, .4967, .9934},
    xticklabels = {$0.07$, $0.29$, $0.48$, , $0.99$},
    ytick = {0.0, 0.0749, 0.6694, 0.9806, 1.},
    yticklabels = {$0.0$, $0.07$, $0.67$, $0.98$, },
    grid = both,
    grid style={gray!30},
    xmin = 0, xmax = 1.05,
    ymin = 0, ymax = 1.05,
    samples = 1000,
    trig format plots = rad,
    declare function = {p(\p,\w) = 1 - \p * (1 - cos(\w * x))/2;}
]
    \addplot[color = blue!80!black, domain=0:1]{p(0.9998, 6.3883)};
    \label{plt:b_0.01}
    \addplot[color = cyan!60!blue, domain=0:1]{p(0.9251, 6.5757)};
    \label{plt:b_0.20}
    \addplot[color = orange!80!red, domain=0:1]{p(0.3306, 11.0)};
    \label{plt:b_1.00}
    \addplot[color = red!80!black, domain=0:1]{p(0.01937, 45.4423};
    \label{plt:b_5.00}
\end{axis}
    \node[anchor = north west] at (rel axis cs:1.02, .731) {
    \begin{tabular}{cc|cccc}
        & $b$ & $\eta^2$ & $P_\text{trap}$ & $\omega$ & $\pi / \omega$ \\
        \hline
        \ref{plt:b_0.01} & $0.01$ & $0.0002$ & $0.0002$ & $\phantom{0}6.3252$ & $0.4967$ \\
        \ref{plt:b_0.20} & $0.20$ & $0.0801$ & $0.0749$ & $\phantom{0}6.5757$ & $0.4778$ \\
        \ref{plt:b_1.00} & $1.00$ & $2.0250$ & $0.6694$ & $11.0000$ & $0.2856$ \\
        \ref{plt:b_5.00} & $5.00$ & $50.6250$ & $0.9806$ & $45.4423$ & $0.0691$ \\
     \end{tabular}
    };
\end{tikzpicture}
    \caption{Return Probabilities $P_{1 \to 1}(t)$ in the case of $A$ being a single-node graph ($m = 1$) and $B$ a clique of size $n = 10$. The inter-subgraph coupling strength is kept constant at $c = 1$, while $b$, the coupling strength within $B$, is varied.}
    \label{fig:coupling_strength}
\end{figure}
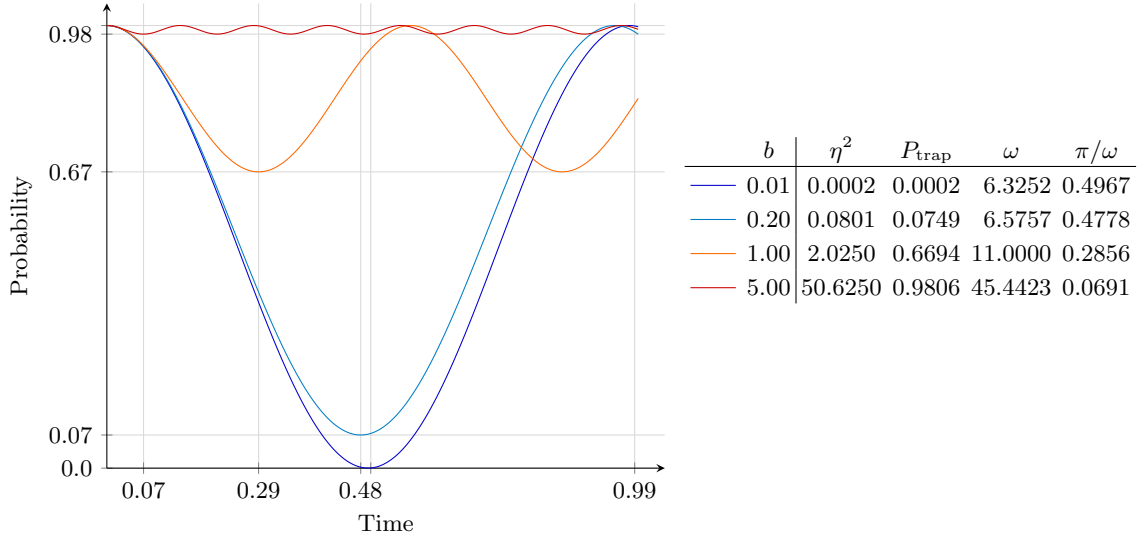
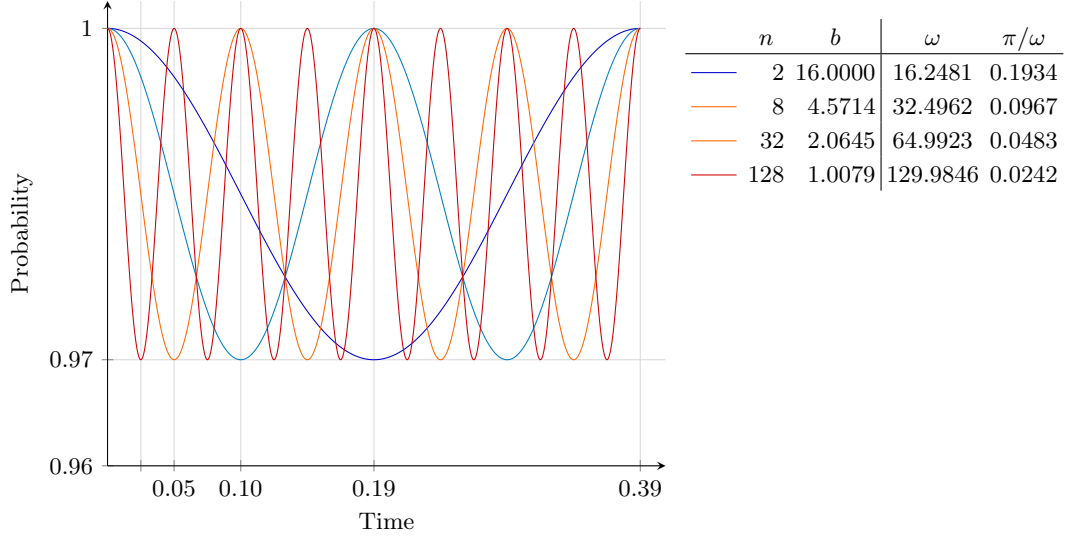
\begin{figure}
    \centering
    \begin{tikzpicture}
\begin{axis}[
    width=.5\linewidth,
    axis lines=left,
    xlabel = {Time},
    ylabel = {Probability},
    xtick = {0.0242, 0.0483, 0.0967, 0.1934, 0.3867},
    xticklabels = {, $0.05$, $0.10$, $0.19$, $0.39$},
    ytick = {0.96, 0.9697, 1.},
    grid = both,
    grid style={gray!30},
    xmin = 0, xmax = 0.405,
    ymin = .96, ymax = 1.0025,
    samples = 1000,
    trig format plots = rad,
    declare function = {p(\p,\w) = 1 - \p * (1 - cos(\w * x))/2;}
]
    \addplot[color = blue!80!black, domain=0:0.3867]{p(0.0303, 16.2481)};
    \label{plt:n_2}
    \addplot[color = cyan!60!blue, domain=0:0.3867]{p(0.0303, 32.4962)};
    \label{plt:n_8}
    \addplot[color = orange!80!red, domain=0:0.3867]{p(0.0303, 64.9923)};
    \label{plt:n_32}
    \addplot[color = red!80!black, domain=0:0.3867]{p(0.0303, 129.9846)};
    \label{plt:n_128}
\end{axis}
    \node[anchor = north west] at (rel axis cs:1.02, .98) {
    \begin{tabular}{ccc|cc}
         & $n$ & $b$ & $\omega$ & $\pi / \omega$ \\
         \hline
         \ref{plt:n_2} & $\phantom{00}2$ & $16.0000$ & $16.2481$ & $0.1934$ \\
         \ref{plt:n_8} & $\phantom{00}8$ & $\phantom{0}4.5714$ & $32.4962$ & $0.0967$ \\
         \ref{plt:n_32} & $\phantom{0}32$ & $\phantom{0}2.0645$ & $64.9923$ & $0.0483$ \\
         \ref{plt:n_128} & $128$ & $\phantom{0}1.0079$ & $129.9846$ & $0.0242$
    \end{tabular}
    };
\end{tikzpicture}
    \caption{Return Probabilities $P_{1 \to 1}(t)$ in the case of $A$ being a single-node graph ($m = 1$) and $B$ a clique. The inter-subgraph coupling strength $c = 1$ and the probability fraction  $\eta^2 = 32$ are kept constant, while the coupling strength $b$ and the size of the clique $n$ are varied.}
    \label{fig:constant_trapping}
\end{figure}
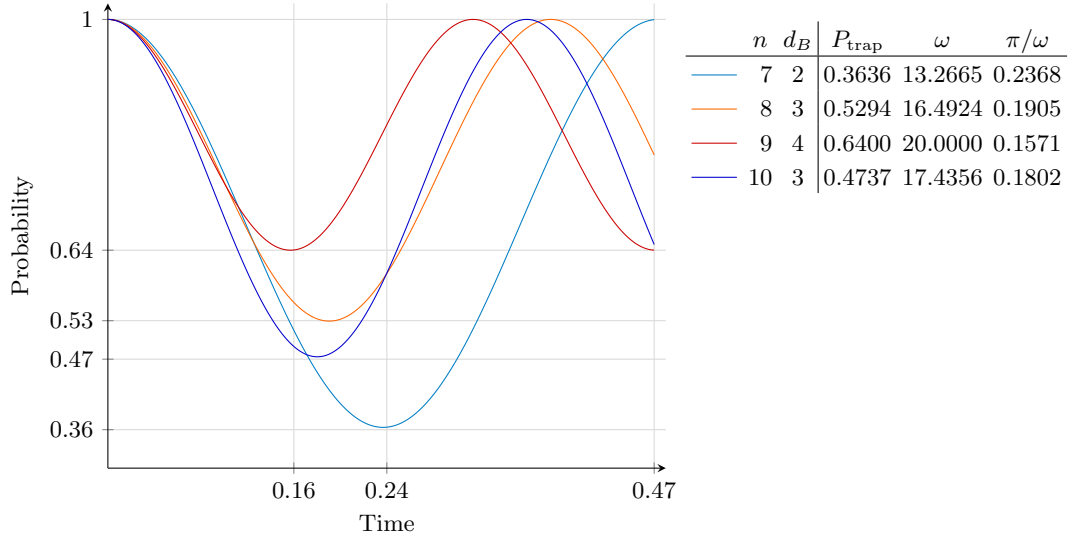
\begin{figure}
    \centering
    \begin{tikzpicture}
\begin{axis}[
    width=.5\linewidth,
    axis lines=left,
    xlabel = {Time},
    ylabel = {Probability},
    xtick = {0.16, 0.24, 0.47},
    ytick = {0.36, 0.47, 0.53, 0.64, 1.},
    grid = both,
    grid style={gray!30},
    xmin = 0, xmax = .48,
    ymin = .3, ymax = 1.025,
    samples = 1000,
    trig format plots = rad,
    declare function = {p(\p,\w) = 1 - \p * (1 - cos(\w * x))/2;}
]
    \addplot[color = cyan!60!blue, domain=0:.47]{p(0.6364, 13.2665};
    \label{plt:n_7}
    \addplot[color = orange!80!red, domain=0:.47]{p(0.4706, 16.4942)};
    \label{plt:n_8}
    \addplot[color = red!80!black, domain=0:.47]{p(0.3600, 20.0000)};
    \label{plt:n_9}
    \addplot[color = blue!80!black, domain=0:.47]{p(0.5263, 17.4356)};
    \label{plt:n_10}
    
\end{axis}
    \node[anchor = north west] at (rel axis cs:1.02, .98) {
    \begin{tabular}{ccc|ccc}
         & $n$ & $d_B$ & $P_\text{trap}$ & $\omega$ & $\pi / \omega$ \\
         \hline
         \ref{plt:n_7} & $\phantom{0}7$ & $2$ & $0.3636$ & $13.2665$ & $0.2368$ \\
         \ref{plt:n_8} & $\phantom{0}8$ & $3$ & $0.5294$ & $16.4924$ & $0.1905$ \\
         \ref{plt:n_9} & $\phantom{0}9$ & $4$ & $0.6400$ & $20.0000$ & $0.1571$ \\
         \ref{plt:n_10} & $10$ & $3$ & $0.4737$ & $17.4356$ & $0.1802$
    \end{tabular}
    };
\end{tikzpicture}
    \caption{Return Probabilities $P_{1 \to 1}(t)$ in the case of $A$ being a single-node graph ($m = 1$) and $B$ a graph of intermediate density with $d_B = \left \lceil{\sqrt{n}}\right \rceil \approx \sqrt{n}$. Coupling strengths $b = 4$, $c = 2$ are kept constant, while the size $n$ is varied.}
    \label{fig:intermediate_graphs}
\end{figure}
\begin{figure}
    \centering
    \begin{tikzpicture}
\begin{axis}[
    width=.5\linewidth,
    axis lines=left,
    xlabel = {Time},
    ylabel = {Probability},
    xtick = {0.0, 0.5, 1.0},
    ytick = {0.25, .5, 0.75, 1.},
    grid = both,
    grid style={gray!30},
    xmin = 0, xmax = 1.05,
    ymin = .0, ymax = 1.05,
    samples = 1000,
    trig format plots = rad,
    declare function = {p(\p,\w) = 1 - \p * (1 - cos(\w * x))/2;}
]
    \addplot[color = red!80!black, domain=0:1.]{p(1.00, 6.2832};
    \label{plt:p_00}
    \addplot[color = orange!80!red, domain=0:1.]{p(0.75, 6.2832)};
    \label{plt:p_25}
    \addplot[color = cyan!60!blue, domain=0:1.]{p(0.50, 6.2832)};
    \label{plt:p_50}
    \addplot[color = blue!80!black, domain=0:1.]{p(0.25, 6.2832)};
    \label{plt:p_75}
    
\end{axis}
    \node[anchor = north west] at (rel axis cs:1.02, .98) {
    \begin{tabular}{ccc|c}
         & $b$ & $c$ & $P_\text{trap}$  \\
         \hline
         \ref{plt:p_00} & $0.0000$ & $1.5708$ & $0.00$ \\
         \ref{plt:p_25} & $1.0472$ & $1.3603$ & $0.25$ \\
         \ref{plt:p_50} & $1.4810$ & $1.1107$ & $0.50$ \\
         \ref{plt:p_75} & $1.8138$ & $0.7854$ & $0.75$
    \end{tabular}
    };
\end{tikzpicture}
    \caption{Return Probabilities $P_{1 \to 1}(t)$ in the case of $A$ being a single-node graph ($m = 1$) and $B$ a clique of size $n = 4$. The angular frequency $\omega = 2 \pi$ is kept constant, while we vary the coupling strengths $b$ and $c$.}
    \label{fig:constant_frequency}
\end{figure}
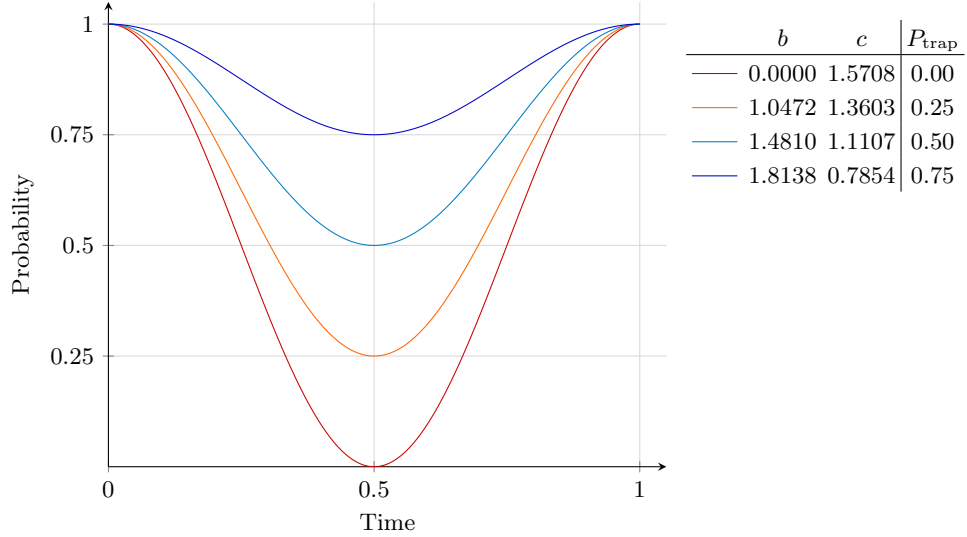

\subsection{A Clique Coupled to an Arbitrary Subgraph}
\label{subsec:clique}
Let us now consider the more demanding situation of Subgraph $A$ being a clique and $B$ an arbitrary regular graph. The subgraph’s weighted adjacency matrix
\begin{align}
\label{eqn:clique_hamiltonian}
G_A & =
\begin{pmatrix}
1 & a & \cdots & a \\
a & 1 & \cdots & a \\
\vdots & \vdots & \ddots & \vdots \\
a & a & \cdots & 1
\end{pmatrix}
\end{align}
has two eigenvalues: The first one is the already-mentioned $\lambda_A = (m - 1) a + 1$ which, together with $\ket{\phi} = \frac{1}{\sqrt{m}} \uv_m$, forms the uniform eigenpair. The second eigenvalue $1 - a$ is degenerate for $m \geq 3$. The corresponding orthonormal eigenvectors are
\begin{align}
\label{eqn:clique_eigenvectors}
\ket{\phi_i} & =
\frac{1}{\sqrt{i(i + 1)}}
(\underbrace{1, \ldots, 1}_i, -i, 0, \ldots, 0),
\end{align}
where $i \in \coll{1, \ldots, m - 1}$. Due to the degeneracy of the second eigenvalue, the general Schrödinger propagator from Equation \eqref{eqn:schroedinger_propagator} simplifies significantly to
\begin{align}
\label{eqn:clique_schroedinger}
U(t) & = e^{-\iu t (1 - a)} \left(M_A - \ket{\psi_A} \bra{\psi_A} \right) +
\left(e^{-\iu t \lambda_+} \ket{\psi_+} \bra{\psi_+}\right) +
\left(e^{-\iu t \lambda_-} \ket{\psi_-} \bra{\psi_-}\right) +
\left(\sum_{k = m + 2}^{m + n} e^{-it \lambda_k} \ket{\psi_k} \bra{\psi_k}\right).
\end{align}
In this section we are interested in the nodal transfer probability, i.e. the probability of measuring the excitation in the $j$-th node, assuming that it was initially at the $i$-th node. Due to the cyclic permutation property of the trace, we have
\begin{align}
\label{eqn:nodal_transfer_probabilites}
P_{i \to j}(t) & =
\Tr \left[\ket{j}\bra{j} U(t) \ket{i} \bra{i} U(t)^\dagger \right] =
\abs{\braket{j | U(t) | i}}^2.
\end{align}
We obtain the nodal transfer probabilities
\begin{align}
\label{eqn:nodal_return_probability}
P_{i \to j}(t) & = \abs{
e^{-\iu t (1 - a)} \left(\delta_{ij} - \frac{1}{m} \right) +
\frac{1}{m} C^2 e^{-\iu t \lambda_+} + \frac{1}{m} S^2 e^{-\iu t \lambda_-}
}^2,
\end{align}
$i, j \in A$, by inserting \eqref{eqn:clique_schroedinger} into \eqref{eqn:nodal_transfer_probabilites} and using the scalar products from Equation \eqref{eqn:scalar_products}.

We begin our comments with the notes on complete depletion which occurs iff 
$\inf_{t} P_{i \to i}(t) = 0$. In contrast to full oscillations, which occur at subgraph level, complete depletion is a node-level effect. In full oscillations an excitation starting in Subgraph $A$ leaves it for Subgraph $B$. In complete depletion, however, an excitation starting in node $i$ will leave it, but its wave function, in general, will split into a part still going to other nodes of $A$, and a part travelling to Subgraph $B$. Pulling out a common phase in Equation \eqref{eqn:nodal_return_probability} and employing the triangle inequality we obtain
\begin{align}
\label{eqn:complete_depletion}
P_{1 \to 1}(t) & = \abs{\left(1 - \frac{1}{m} \right) +
\frac{1}{m} C^2 e^{-\iu t (\lambda_+ - (1 - a))} +
\frac{1}{m} S^2 e^{-\iu t (\lambda_- - (1 - a))}} \geq 1 - \frac{2}{m},
\end{align}
which shows that complete depletion cannot occur for $m \geq 3$. With increasing size of the clique $A$, the excitation is even trapped in its original position uniformly in all other graph parameters.

The case of a single-node subgraph has been discussed in Section \ref{subsec:single-node}: for $m = 1$ full oscillations and complete depletion coincide. We therefore focus on the case $m = 2$. To that end let us define $\omega_\pm \coloneqq \lambda_\pm - (1 - a)$ and recall some basic facts of quasi-periodic time evolution: For one, if $\frac{\omega_+}{\omega_-} = \frac{2l + 1}{2k + 1}$ for some $k, l \in \Intgrs$ (odd-over-odd rational), we observe complete depletion in finite time. More precisely, for $t_\text{d} = \pi \frac{2k + 1}{\omega_-}$ one has $P_{i \to i}(t_\text{d}) = 0$ and this situation repeats with period $2 \pi \frac{2k + 1}{\omega_-}$. For the other, if $\frac{\omega_+}{\omega_-} \in \Reals \setminus \Rationals$ (irrational fraction) the trajectory $t \mapsto t (\omega_+, \omega_-)$ is dense in the two-dimensional torus $\mathbb{T}^2$ and we again observe complete depletion but the infimum is never attained. In all other cases (even-over-odd or odd-over-even rational), $P_{i \to i}(t)$ is bounded away from zero, i.e. complete depletion does not occur.

At a time $t_\text{d}$ at which a (nearly) complete depletion occurs, both complex exponentials in Equation \eqref{eqn:complete_depletion}, $e^{-\iu t_\text{d} \omega_{\pm}}$, are close to $-1$. Hence, by Equation \eqref{eqn:nodal_return_probability}, we have
\begin{align}
\label{eqn:nodal_transfer_probability}
P_{1 \to 2}(t_\text{d}) & = 
\abs{\frac{1}{2} C^2 e^{-\iu t_\text{d} \omega_+} +
\frac{1}{2} S^2 e^{-\iu t_\text{d} \omega_-} - \frac{1}{2}}^2 \approx 1,
\end{align}
i.e. the excitation moves (almost) fully to the second node and not into Subgraph $B$.

Finally, let us study an asymptotic behavior for the case $m = 2$. We fix all coupling strengths, the family of Subgraph $B$, i.e. $d_B = d_B(n)$, and let it grow in size $n \to \infty$. Moreover, we assume that the graph density approaches the limit
\begin{align}
\label{eqn:density}
\alpha & \coloneqq \lim_{n \to \infty} \frac{d_B(n)}{n}.
\end{align}
Then, asymptotically complete trapping implies asymptotically complete depletion, unless the balance condition
\begin{align}
\label{eqn:balance}
\alpha & = \frac{c^2}{a b} 
\end{align}
is satisfied. To see this, let us recall from Section \ref{subsec:trapping_subgraphs} that asymptotically complete trapping is equivalent to $\eta \to -\infty$, because Subgraph $A$ has the fixed size $m = 2$. In the case of asymptotically complete trapping, Equations \eqref{eqn:perturbed_eigenvectors} and \eqref{eqn:scalar_products} imply that $C^2 \approx 0$ and thus the corresponding term is negligible in Equation \eqref{eqn:complete_depletion}. Therefore, we indeed obtain asymptotically complete depletetion, unless the remaining frequency $\omega_- = \lambda_- - (1 - a)$ vanishes. In fact, expanding definitions and gathering terms yields
\begin{align*}
\omega_- & =
\frac{a}{2} \left( 3 + \sqrt{1 + \frac{1}{\eta^2}} \right) +
\frac{b}{2} d_B \left(1 - \sqrt{1 + \frac{1}{\eta^2}} \right),
\end{align*}
which, after employing the difference of squares formula, becomes
\begin{align*}
\frac{a}{2} \left( 3 + \sqrt{1 + \frac{1}{\eta^2}} \right) -
\frac{4 c^2}{b} \frac{n}{d_B} \frac{1}{\left(1 - \frac{a}{b d_B}\right)^2} \frac{1}{1 + \sqrt{1 + \frac{1}{\eta^2}}} & \xrightarrow{n \to \infty}
2 \left(a - \frac{c^2}{b \alpha} \right)
\end{align*}
and explains Condition \eqref{eqn:balance}. Let us conclude by remarking that the above generalizes the missing-link case from \cite{Buridan_paper} to the situation where the link is not fully missing ($a \neq 0$) and is not necessarily a part of a fully connected graph.

\section{Balanced Mode Coupling and Further Generalization}
\label{sec:generalization}
Our choice of the model’s structure was as to generalize significantly the previously considered examples while still maintaining a sufficiently high degree of symmetry to perform explicit calculations. More precisely, we have chosen the block structure of the weighted adjacency matrix \eqref{eqn:adjacency_matrix} to generalize from “mind states” (nodes) to “attitudes” (blocks) in the Buridan’s ass interpretation. We started with the simplest non-trivial case of two blocks. The uniform coupling of all nodes in Subgraph $A$ to all nodes in Subgraph $B$ with coupling strength $c$ allowed us to treat earlier examples as special cases of the current model. The assumption of regularity of the subgraphs had two important implications. For one, it provided us with the explicit knowledge of an eigenvector-eigenvalue pair for each subgraph (uniform eigenpair). For the other, it was compatible with the uniform coupling structure and gave rise to the invariant Hilbert spaces $\mathcal{H}_{A'}$, $\mathcal{H}_{B'}$. This resulted in the two-level dynamics on the level of subgraphs presented in Section \ref{sec:subgraph_effects} and \ref{subsec:single-node}.

In fact, we can easily replace the assumption of the regularity of the subgraphs by its two implications, if we maintain a compatible coupling structure. To that end, let us consider the Hamiltonian
\begin{align}
\label{eqn:generalized_hamiltonian}
H & =
H_A + H_B +
\tilde{c}
\left(\ket{\psi_A} \bra{\psi_B} + \ket{\psi_B} \bra{\psi_A} \right)
\end{align}
corresponding to the block structure from Equation \eqref{eqn:adjacency_matrix}. Above, $H_A$ ($H_B$) is the Hamiltonian corresponding to the matrix $G_A$ ($G_B$, respectively). Instead of the uniform eigenpair for $A$ ($B$) we now consider an arbitrary normalized eigenvector $\ket{\psi_A}$ ($\ket{\psi_B}$) with corresponding eigenvalue $\lambda_A$ ($\lambda_B$, respectively). In the following we will transfer our findings from Section \ref{sec:subgraph_effects} to this more general setting.

To have a concrete example in mind, consider $A$ ($B$) being a clique with the intra-graph coupling strength $a$ ($b$, respectively) and
\begin{equation}
\label{eqn:generalized_coupling}
\begin{aligned}
\lambda_A & = 1 - a, & \lambda_B & = 1 - b, \\
\ket{\psi_A} & = \frac{1}{\sqrt{2}} \left( \ket{1} - \ket{2} \right), &
\ket{\psi_B} & = \frac{1}{\sqrt{2}} \left( \ket{m + 1} - \ket{m + 2} \right).
\end{aligned}
\end{equation}
Let us note that in the site basis $(\ket{i})_i$ this particular example results in a rather sparse coupling matrix $\ket{\psi_A} \bra{\psi_B}$ with only four non-zero entries, whereas our initial coupling matrices were dense. We also point out that the high degeneracy of the eigenvalues in the clique setting allows for considerable freedom in the choice of the coupling eigenvectors: We can choose any normalized $\ket{\psi_A}$ ($\ket{\psi_B}$) orthogonal to the uniform state in $A$ ($B$, respectively), i.e. any normalized vector whose entries sum up to zero in the site basis. We shall refer to this case as the balanced mode coupling.

We shall not assume any particular choice of the coupling states $\ket{\psi_A}$ and $\ket{\psi_B}$, unless stated explicitly. We keep the assumption of $A$ and $B$ having disjoint supports and complement $\ket{\psi_A}$ with the eigenvectors $\ket{\psi_1}, \ldots, \ket{\psi_{m - 1}}$ of $H_A$ to an orthonormal basis of $\mathcal{H}_A$. The corresponding procedure for Subgraph $B$ yields $\ket{\psi_{m + 2}}, \ldots, \ket{\psi_{m + n}}$.  The dependence of the true degrees of freedom $\eta$ and $\omega$ on the graph parameters needs to be replaced with the dependence on the generalized graph parameters $\ket{\psi_A}$, $\ket{\psi_B}$, and $\tilde{c}$. For the balanced mode coupling, as in Display \eqref{eqn:generalized_coupling}, for example, we obtain
\begin{align}
\label{eqn:generalized_dependence}
\eta & = \frac{b - a}{2 \tilde{c}}, & \omega & = \sqrt{4 \tilde{c}^2 - (b - a)^2}.
\end{align}
As in the uniform coupling case, the original eigenvectors $\ket{\psi_A}$ and $\ket{\psi_B}$ of the subgraphs’ Hamiltonians are rotated by the angle $\varphi = \varphi(\eta)$. The resulting eigenvectors $\ket{\psi_+}$ and $\ket{\psi_-}$, on the one hand, are the only means of propagation between the subgraphs, and the strength of this propagation path depends on the rotation angle $\varphi$. The original eigenvectors $\ket{\psi_1}, \ldots, \ket{\psi_{m - 1}}$, however, remain unchanged by the coupling and their span, $\mathcal{H}_{A'}$, is invariant under the time evolution. The same is true for the original eigenvectors $\ket{\psi_{m + 1}}, \ldots, \ket{\psi_{m + n}}$ and their span, $\mathcal{H}_{B'}$. For an excitation inserted into Subgraph $A$ it means that its part supported in $\mathcal{H}_{A'}$ will stay there forever. Its other part, aligned with $\ket{\psi_A}$, is further subdivided into two contributions. The first contribution will stay in state $\ket{\psi_A}$, while the second contribution will oscillate fully between $\ket{\psi_A}$ and $\ket{\psi_B}$. This is best summarized by rewriting the return probability from Equation \eqref{eqn:return_probability_concrete} as
\begin{align}
\label{eqn:return_probability_three_contributions}
P_{A \to A}(\rho, t) & =
\Tr[M_{A'} \rho] +
P_\text{trap} \braket{\psi_A | \rho | \psi_A} +
P_\text{trans} \frac{1 + \cos(\omega t)}{2} \braket{\psi_A | \rho | \psi_A},
\end{align}
where we used $1 = \Tr[\rho] = \Tr[M_A \rho] = \Tr[(M_{A'} + \ket{\psi_A}\bra{\psi_A}) \rho]$. The first term corresponds to the first part of the excitation, while the second and third terms correspond to the two contributions of the second part. The maximal transition probability $P_\text{trans}$ and the minimal trapping probability $P_\text{trap} = 1 - P_\text{trans}$, i.e. the splitting of the second part into the two contributions, are complementary and depend only on the generalized graph parameters. More precisely, they are a function of the single parameter $\eta^2 = \frac{P_\text{trap}}{P_\text{trans}}$, which is invariant under swapping of the subgraphs. The angular frequency of the oscillations $\omega$ also depends only on the generalized graph parameters and is also invariant under swapping of the subgraphs. However, $P_\text{trans}$ (or, equivalently, $\eta$) alone does not fix $\omega$, so they can be set independently. This fact is illustrated in Figure \ref{fig:minimal_return_probability} depicting the minimal return probability $P_{A \to A}(\ket{\psi_A} \bra{\psi_A}, t)$ for different choices of the geometric parameter $\eta$ and angular frequency $\omega$.
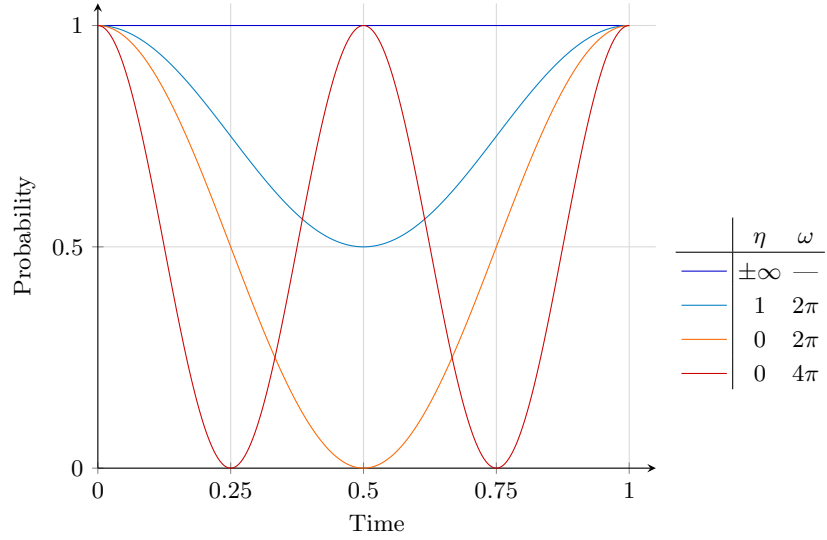
\begin{figure}
    \centering
    \begin{tikzpicture}
\begin{axis}[
    width=.5\linewidth,
    axis lines=left,
    xlabel = {Time},
    ylabel = {Probability},
    xtick = {0., .25, .5, .75, 1.},
    ytick = {0., .5, 1.},
    grid = both,
    grid style={gray!30},
    xmin = 0, xmax = 1.05,
    ymin = 0, ymax = 1.05,
    samples = 1000,
    trig format plots = rad,
    declare function = {p(\p,\w) = 1 - \p * (1 - cos(\w * x))/2;}
]
    \addplot[color = blue!80!black, domain=0:1]{p(.0, 2*pi)};
    \label{plt:frac_infty_omega_0}
    \addplot[color = cyan!60!blue, domain=0:1]{p(.5, 2*pi)};
    \label{plt:frac_1_omega_2_pi}
    \addplot[color = orange!80!red, domain=0:1]{p(1., 2*pi)};
    \label{plt:frac_0_omega_2_pi}
    \addplot[color = red!80!black, domain=0:1]{p(1., 4*pi};
    \label{plt:frac_0_omega_4_pi}
\end{axis}
    \node[anchor = north west] at (rel axis cs:1.02, .5579) {
    \begin{tabular}{c|cc}
        & $\eta$ & $\omega$ \\
        \hline
        \ref{plt:frac_infty_omega_0} & $\pm \infty$ & --- \\
        \ref{plt:frac_1_omega_2_pi} & $1$ & $2 \pi$ \\
        \ref{plt:frac_0_omega_2_pi} & $0$ & $2 \pi$ \\
        \ref{plt:frac_0_omega_4_pi} & $0$ & $4 \pi$ \\
     \end{tabular}
    };
\end{tikzpicture}
    \caption{Minimal Return Probability $P_{A \to A}\left(\ket{\psi_A} \bra{\psi_A}, t\right)$ for different values of the geometric parameter $\eta$ and oscillation frequency $\omega$. Higher $\abs{\eta}$ correspond to higher minimal trapping. For $\abs{\eta} = \infty$ complete trapping occurs (violet curve). For $\eta = 0$ there are full oscillations. Moreover, the oscillation frequency is independent of the minimal trapping probability: Red curve’s frequency is twice that of the orange curve, yet both exhibit full oscillations.}
    \label{fig:minimal_return_probability}
\end{figure}

Full oscillations between subgraphs occur exactly when the initial value $\rho$ is $\ket{\psi_A} \bra{\psi_A}$ and $P_\text{trans} = 1$. In the balanced mode case, as in Display \eqref{eqn:generalized_coupling}, this occurs exactly for $a = b$. Let us emphasize that in contrast to the uniform coupling case ($(m - 1) a \overset{!}{=} (n - 1) b$) this condition is independent of the sizes of the subgraphs. Minimal trapping in the subgraph occurs also if and only if $\rho = \ket{\psi_A} \bra{\psi_A}$ and the minimal value is $P_\text{trap}$ in this case. For the balanced mode coupling, minimal trapping can be complete only asymptotically, either for large eigenvalue difference $b - a$, or for small inter-graph coupling $\tilde{c}$.

Finally, if we average over initial states, we see that the excitation becomes fully trapped with diverging $m$ universally in the other generalized graph parameters. While this fact might seem surprising at first glance it becomes clear from the discussion above --- from $m$ orthogonal basis functions of $\mathcal{H}_A$ only one leads to propagation to Subgraph $B$, so its contribution vanishes with growing $m$. 

\section{Conclusions}
\label{sec:conclusions}
We proposed a model generalizing the previously studied cases of a star graph, a fully connected graph, and a fully connected graph with one edge missing. In our detailed study of two fully interconnected regular graphs we focused on the possibility of transferring an excitation from one subgraph to the other. We observed effects at the subgraph- and at the node-level.

Full oscillations and complete trapping arise at subgraph level and are due to the underlying two-level dynamics. Full oscillations require a careful but simple balance between the subgraphs’ degrees and the coupling strengths. Additionally, the excitation needs to be initially in the uniform state of one of the subgraphs. Asymptotically complete minimal trapping requires an interplay of two independent mechanisms: The size of the subgraphs and their local structure, i.e. their degree and coupling strengths. More precisely, if one of the subgraphs is dense, then the inhibition of transfer will be the more pronounced the bigger this dense component becomes. If we average over all initial states localized in one of the subgraphs, we observe asymptotic average complete trapping: The probability of finding the excitation in the initial subgraph approaches one uniformly in size and structure of the other subgraph.

We might look at the problem studied as an example of opinion spreading between two groups represented by two graphs. Though they are highly coupled the transfer of opinions happens only in certain regimes and sometimes even the size is not the main control parameter for the transfer. The fundamental effect behind this fact is interference and underlines once again that unexpected effects can occur as its consequence. It is also worth returning to the classical Buridan's ass paradox that motivated this work \cite{Leibniz, Aristotle, AI}. There, symmetry between two equally attractive options is what traps the decision-maker in indecision. Our results show the opposite: the symmetry condition $d_A a = d_B b$ from Equation \eqref{eqn:full_oscillations_condition} is what produces full oscillation between the opinion bubbles, while it is asymmetry --- one bubble growing large and dense relative to the other --- that produces trapping. Due to quantum interference, symmetry may enable transitions rather than freezing them.

On the node-level we discussed, for one, the implications of subgraph-level effects when specialized to single-node subgraphs. In particular, we observed a critical subgraph density required for asymptotically complete trapping. For the other, we considered coupling a clique to an arbitrary subgraph. In this situation we could study complete depletion. We proved that complete depletion can only occur if the clique is trivial, i.e. consists of at most two nodes. The two-node subgraph case exhibits the more involved three-level dynamics. In particular, asymptotically complete subgraph trapping implies asymptotically complete depletion, unless the graph density is carefully balanced against the coupling strengths.

Finally, we further generalized our model to coupling of arbitrary two eigenvalue-eigenvector pairs from each subgraph. This allowed us to study the coupling of balanced modes of two cliques identifying exact conditions for full oscillations, asymptotically complete trapping in the subgraphs and observing mean asymptotic trapping.

We find it worth to consider further generalizations of our model. It is intriguing to consider more than two simple coupled graphs. Our previous experience suggests that a certain type of reduction might be possible hence certain typical graph properties control the dynamics of the networks despite the fact that they appear rather complex and complicated at first glance. This might shed light onto quantum analogs of population dynamics as well as the functioning of certain mental processes, assuming a wave-quantum mechanism is relevant. 

\section*{Acknowledgments}
I. J. and J. N. have been supported by the Czech Science Foundation (GA\v CR) project number 26-23973S. 
S. M. B. thanks the Royal Society for the award of a Research Professorship (RP150122 and RSRP/R/210005). T. K. was supported by the Ministry of Culture and Innovation and the National Research, Development and Innovation Office of Hungary (Grant No. TKP-2021-NVA-04). L. C. gratefully acknowledges the support by EU and Czech Ministry of Education under project \verb|CZ.02.01.01/00/22_011/0008569| “Czech Technical University --- International Postdoctoral Programme”.

\appendix
\section{Some Lemmas}
\begin{lemma}
\label{lem:convex_max}
Let $\mathcal{H}$ be a Hilbert space with $\dim(\mathcal{H}) \eqqcolon n < \infty$ and $\psi \in \mathcal{H}$ with $\braket{\psi | \psi} = 1$. Furthermore, let $\mathcal{D}_n$ denote the space of density matrices over $\mathcal{H}$. Then the convex maximization problem
\begin{align}
\label{eqn:convex_max}
\max_{\rho \in \mathcal{D}_n} \Tr[\ket{\psi} \bra{\psi} \rho]
\end{align}
has the unique solution $\rho_0 = \ket{\psi} \bra{\psi}$ and the maximal value is $1$.
\end{lemma}
\begin{proof}
Consider an arbitrary $\rho \in \mathcal{D}_n$. By the spectral theorem it has a representation
\begin{align}
\label{eqn:spectral_representation}
\rho & = \sum_{i = 1}^n \alpha_i \ket{\psi_i} \bra{\psi_i},
\end{align}
with $\sum_{i = 1}^n \alpha_i = 1$, $\alpha_1, \ldots, \alpha_n \geq 0$, and orthonormal $\ket{\psi_1}, \ldots, \ket{\psi_n}$. Then, by the cyclic permutation property of the trace, and the Cauchy-Schwarz inequality,
\begin{align}
\label{eqn:cauchy-schwarz_argument}
\Tr[\ket{\psi} \bra{\psi} \rho] = \sum_{i = 1}^n \alpha_i \abs{\braket{\psi | \psi_i}}^2 \leq \sum_{i = 1}^n \alpha_i \abs{\braket{\psi | \psi}}^2 \abs{\braket{\psi_i | \psi_i}}^2= 1.
\end{align}
Hence $1$ is an upper bound for Problem \eqref{eqn:convex_max}. Clearly, this upper bound is attained for $\rho = \ket{\psi} \bra{\psi}$ and so $\rho_0$ is a maximizer.

Conversely, if $\rho$ is a maximizer, i.e. $\Tr[\ket{\psi} \bra{\psi} \rho] = 1$, Inequality \eqref{eqn:cauchy-schwarz_argument} implies that for any $i$ for which $\alpha_i \neq 0$ we have
\begin{align*}
\abs{\braket{\psi|\psi_i}} & = \N{\psi} \N{\psi_i}.
\end{align*}
Hence, by the Cauchy-Schwarz Theorem, $\psi_i = e^{\iu \theta_i} \psi$ for some $\theta_i \in [-\pi, \pi)$. Due to the orthonormality of $\ket{\psi_1}, \ldots, \ket{\psi_n}$, there must be exactly one $i$ with $\alpha_i \neq 0$ and, due to the sum condition, $\alpha_i = 1$. Thus, indeed, $\rho = e^{\iu (\theta_i - \theta_i)} \ket{\psi} \bra{\psi} = \ket{\psi} \bra{\psi}$.
\end{proof}

\begin{lemma}
\label{lem:unitary_invariant_mean}
Let $\mathcal{D}_n$ be the space of $n \times n$ density matrices and let $\mu$ be a unitary invariant probability measure on $\mathcal{D}_n$, i.e.
\begin{align*}
\mu(\rho) & = \mu(U \rho U^\dagger)
\end{align*}
for any unitary $U$. Then the average density matrix $\overline{\rho}$ w.r.t. to $\mu$ is given by $\overline{\rho} = \frac{1}{n} \ident_n$.
\end{lemma}
\begin{proof}
By the unitary invariance of $\mu$ and the linearity of the integral we immediately obtain that, for any unitary $U$,
\begin{align*}
\overline{\rho} & = \int \rho \mu(\d \rho) = \int U \rho U^\dagger \mu(\d \rho) =
U \int \rho \mu(\d \rho) U^\dagger = U \overline{\rho} U^\dagger,
\end{align*}
i.e. that $\overline{\rho}$ is unitarily invariant. Let $\lambda_1, \ldots, \lambda_n$ be the eigenvalues of $\overline{\rho}$ and $\ket{\psi_1}, \ldots, \ket{\psi_n}$ be their corresponding orthonormal eigenvectors. Consider any two eigenvalues of $\overline{\rho}$, w.l.o.g. $\lambda_1$ and $\lambda_2$, and define the unitary
\begin{align*}
U & \coloneqq
\ket{\psi_1} \bra{\psi_2} + \ket{\psi_2} \bra{\psi_1} + \sum_{k = 3}^n \ket{\psi_k} \bra{\psi_k}.
\end{align*}
Then, by the above,
\begin{align*}
\lambda_1 \ket{\psi_1} & =
\overline{\rho} \ket{\psi_1} =
U \overline{\rho} U^\dagger \ket{\psi_1} =
U \overline{\rho} \ket{\psi_2} =
\lambda_2 U \ket{\psi_2} =
\lambda_2 \ket{\psi_1}.
\end{align*}
As $\ket{\psi_1} \neq 0$, we must have $\lambda_1 = \lambda_2$. This shows that all eigenvalues of $\overline{\rho}$ are equal to, say, $C$, and hence $\overline{\rho} = C \ident_n$. Finally, due to the trace condition $\Tr[\overline{\rho}] = C n = 1$, we indeed have that $\overline{\rho} = \frac{1}{n} \ident_n$.
\end{proof}

\begin{lemma}[Handshake Lemma]
\label{lem:handshake}
Let $G$ be an irreflexive undirected (symmetric) regular graph of $n$ vertices, each of degree $d$. Then, $d n$ is even. Conversely, if $d \in \mathbb{N}_0$, $n \in \mathbb{N}$, are such that $d < n$ and $d n$ is even, then there is an irreflexive undirected regular graph $G$ with $n$ vertices each of degree $d$.
\end{lemma}
\begin{proof}
The first claim immediately follows from the observation that $d n$ is the sum of the degrees of the individual nodes and is hence twice the number of edges (each edge contributes one to the degree of both nodes it connects).

For the second claim, let us assume that $d, n \in \mathbb{N}$ such that $d < n$ and $d n$ is even. We will construct an irreflexive undirected regular graph $G$ with $n$ vertices each of degree $d$ as a certain circulant graph. To that end let us number the edges $0, \ldots, n - 1$ and recall that a graph is circulant, iff there exists a finite binary sequence $(b_0, \ldots, b_{n - 1})$ such that
\begin{align*}
a_{i j} & = b_{j - i \mod n} \in \coll{0, 1}
\end{align*}
for all $i, j \in \coll{0, \ldots, n - 1}$. Let us note that here we deviate from the previous notation and number the indices of the adjacency matrix $(a_{ij})$ of $G$ starting from $0$ to be compatible with modular arithmetics.

The irreflexivity condition, $a_{ii} = 0$ for all $i \in \coll{0, \ldots, n - 1}$, is equivalent to $b_0 = 0$. The $d$-regularity condition, $\sum_{j = 0}^{n - 1} a_{i j} = d$ for all $i \in \coll{0, \ldots, n - 1}$, boils down to
\begin{align*}
\sum_{k = 0}^{n - 1} b_k & = d.
\end{align*}
The symmetry condition, $a_{i j} = a_{j i}$ for all $i, j \in \coll{0, \ldots, n - 1}$, reads
\begin{align}
\label{eqn:symmetry}
b_k & = b_{n - k}
\end{align}
for all $k \in \coll{1, \ldots, n - 1}$. To implement it, we distinguish two cases: If $d$ is even, i.e. $d = 2d'$ for some $d' \in \mathbb{N}_0$, we just set
\begin{align*}
(b_k) & =
(0, \underbrace{1, \ldots, 1}_{d'}, \underbrace{0, \ldots, 0}_{n - 1 - d}, \underbrace{1, \ldots, 1}_{d'}).
\end{align*}
If $d$ is odd, i.e. $d = 2d' + 1$ for some $d' \in \mathbb{N}_0$, then $n$ must be even, i.e. $n - 1 = 2n' + 1$ for some $n' \in \mathbb{N}_0$. In this case we set
\begin{align}
\label{eqn:odd}
(b_k) & =
(0, \underbrace{1, \ldots, 1}_{d'}, \underbrace{0, \ldots, 0}_{n' - d'}, 1, \underbrace{0, \ldots, 0}_{n' - d'}, \underbrace{1, \ldots, 1}_{d'}).
\end{align}
Observing that the (potentially) isolated one in Equation \eqref{eqn:odd} is at index $n' + 1 = n - (n' + 1)$ shows that the symmetry condition \eqref{eqn:symmetry} is sasitfied and finishes the proof.
\end{proof}

\begin{remark}
\label{rem:handshake}
The construction in Lemma \ref{lem:handshake} has a simple visual description: Let us arrange the $n$ vertices uniformly on a circle. For each vertex, looking from the vertex in the direction of the origin, connect it to its  $d' \coloneqq \lfloor \frac{d}{2} \rfloor$ nearest neighbours lying to the left and to its $d'$ nearest neighbours lying to the right. If $d$ is odd, then there is one more edge to distribute. In that case, however, $n$ is even and so there is another vertex lying directly opposite, which we additionally connect to.
\end{remark}

\end{document}